\pdfoutput=1
\documentclass{sigplanconf} 

\usepackage{amssymb}
\usepackage{amsthm}
\usepackage{array}
\usepackage{booktabs}
\usepackage{fancyvrb}
\usepackage{flushend}
\usepackage{mathtools}
\usepackage{microtype}
\usepackage{rgalg}
\usepackage{thmtools,thm-restate}
\usepackage{tikz}
\usepackage{xcolor}

\include{pyg}

\newcommand\longshort[2]{#1}

\usepackage{hyperref}
\definecolor{darkred}{rgb}{0.3,0,0}
\definecolor{darkgreen}{rgb}{0,0.3,0}
\definecolor{darkblue}{rgb}{0,0,0.3}
\hypersetup{colorlinks,linkcolor=darkblue,citecolor=darkblue,urlcolor=darkblue}
\longshort{
\hypersetup{
  pdftitle={Java Generics are Turing Complete},
  pdfauthor={Radu Grigore}}}{}

\newcommand*\<{\mathrel{\blacktriangleleft}}
\renewcommand*\>{\mathrel{\blacktriangleright}}

\newcommand*\cE{{\mathcal E}}
\newcommand*\cL{{\mathcal L}}

\newcommand*\cT{{\mathcal T}}
\newcommand*\cyk{{\it cyk}}
\newcommand*\arity{{\it arity}}
\newcommand*\defeq{\mathrel{:=}}
\newcommand*\df[1]{\emph{#1}}

\newcommand*\inh{\mathrel{<::}}
\newcommand*\NN{\mathbb{N}}
\newcommand*\nullable{{\it nullable}}

\newcommand*\rep[1]{\lfloor#1\rfloor}

\newcommand*\vars{{\it vars}}
\newcommand*\subtype{\mathrel{<:}}

\newcommand*\dL{{\sf L}}
\newcommand*\dR{{\sf R}}
\newcommand*\dS{{\sf S}}

\DeclareMathOperator\dom{dom}

\declaretheoremstyle[notebraces={}{},bodyfont=\itshape]{nobraces}

\declaretheorem[style=nobraces]{theorem}

\declaretheorem[sibling=theorem,style=definition,qed=\qedsymbol]{example}
\declaretheorem[sibling=theorem,style=plain]{proposition}
\declaretheorem[sibling=theorem,style=remark,qed=\qedsymbol]{remark}

\begin{document}
\toappear{}





\title{Java Generics are Turing Complete}

\authorinfo{\longshort{Radu Grigore}{}}
           {\longshort{University of Kent, United Kingdom}{}}
           {}

\maketitle

\begin{abstract} 
This paper describes a reduction
  from the halting problem of Turing machines to subtype checking in Java.
It follows that subtype checking in Java is undecidable,
  which answers a question posed by Kennedy and Pierce in 2007.
It also follows that Java's type checker can recognize any recursive language,
  which improves a result of Gil and Levy from~2016.
The latter point is illustrated by a parser generator for fluent interfaces.
\end{abstract}



\section{Introduction} 

Is Java type checking decidable?
This is an interesting theoretical question,
  but it is also of interest to compiler developers
  \cite{kotlin-call}.
Since Java's type system is cumbersome for formal reasoning,
  several approximating type systems have been studied.
Two of these type systems are known to be undecidable:
  \cite{kennedy2007} and \cite{wehr2009}.
Alas, neither of these is a subset of Java's type system:
  the reduction of \cite{kennedy2007} uses multiple instantiation inheritance,
    which is forbidden in Java;
  the reduction of \cite{wehr2009} uses bounded existential types
    in a way that goes beyond what Java wildcards can express.
So, neither result carries over to Java.

Not knowing a decidability proof was enough to spur interest into decidable fragments.
In \cite{kennedy2007},
  certain `recursive--expansive' patterns in the inheritance are forbidden;
  this solution was adopted in Scala,
    whose typechecking remains undecidable for other reasons.
In \cite{taming-wildcards},
  some restrictions on the use of `\hbox{\tt ? super}' are imposed.
In \cite{f-bounded},
  one distinguishes between classes and `shapes'.
Here, we shall prove that Java's type system is indeed undecidable,
  which justifies further these existing studies.

In a separate line of work,
  \cite{det-cfg-java} show that Java's type system
    can recognize deterministic context free languages,
  thus giving the first nontrivial lower bound directly applicable
    to Java's type system.
Our main result immediately implies that
  Java's type system can recognize any recursive language.
In addition, we shall show that
  Java's type system can recognize context free language in polynomial time.
More precisely, the two results are the following:

\begin{restatable}{theorem}{thmmain}
  \label{thm:main}
It is undecidable whether $t \subtype t'$ according to a given class table.
\end{restatable}

\begin{restatable}{theorem}{thmfluent}\label{thm:fluent}
Given is a context free grammar~$G$
  that describes a language $\cL\subseteq\Sigma^*$
  over an alphabet~$\Sigma$ of method names.
We can construct
  Java class definitions,
  a type~$T$,
  and expressions ${\it Start}$, ${\it Stop}$
such that the code
\[ T\; \ell \quad=\quad
  {\it Start}.f^{(1)}().f^{(2)}()\ldots f^{(m)}().{\it Stop} \]
type checks if and only if $f^{(1)}f^{(2)}\ldots f^{(m)} \in \cL$.
\textcolor{darkgreen}{%
Moreover,
  the class definitions have size polynomial in the size of~$G$,
  and the Java code can be type-checked in time polynomial in the size of~$G$.
}
\end{restatable}

\autoref{thm:main}
  is proved by a reduction from the halting problem of Turing machines
  to subtype checking in Java (\autoref{sec:result}).
The proof is preceded
  by an informal introduction to Java wildcards (\autoref{sec:wildcards})
  and by some formal preliminaries
  (Sections \ref{sec:prelim}~and~\ref{sec:subtype-machine}).
It is followed by \autoref{thm:fluent},
  which is an application to generating parsers for fluent interfaces
  (\autoref{sec:fluent}).
The parser generator makes use of a compiler
  from a simple imperative language into Java types;
  this compiler is described next (\autoref{sec:compiler}).
Before we conclude,
  we reflect on the implications of
  Theorems \ref{thm:main}~and~\ref{thm:fluent} (\autoref{sec:discussion}).

\section{Background: Java Wildcards}\label{sec:wildcards} 

This section introduces Java generics, wildcards, and their bounds, by example.
The presentation is necessarily incomplete.
For the definitive reference, see \cite{java-ls}.

Java generics are used, for example, to implement lists.
When implementing the list, its elements are given a generic type.
When using the list, the type of its elements is fixed.
The interaction between generics and subtyping is not trivial.
${\tt List}\langle{\tt Number}\rangle$
  is not a subtype of ${\tt List}\langle{\tt Integer}\rangle$:
  if it were then we could extract integers from a list of numbers;
nor is ${\tt List}\langle{\tt Integer}\rangle$
  is a subtype of ${\tt List}\langle{\tt Number}\rangle$:
  if it were then we could insert numbers in a list of integers.

Intuitively, there is a difference between extracting and inserting data.
For example, suppose we want to implement a function ${\tt firstNum}$
  that extracts the first number in a list.
The implementation should work for both
  ${\tt List}\langle{\tt Integer}\rangle$
  and
  ${\tt List}\langle{\tt Number}\rangle$.
We could use bounded generics
{\small\begin{Verbatim}[commandchars=\\\{\}]
\PY{o}{\PYZlt{}}\PY{n}{x} \PY{k+kd}{extends} \PY{n}{Number}\PY{o}{\PYZgt{}} \PY{n}{Number} \PY{n}{firstNum}\PY{o}{(}\PY{n}{List}\PY{o}{\PYZlt{}}\PY{n}{x}\PY{o}{\PYZgt{}} \PY{n}{xs}\PY{o}{)}
\PY{o}{\PYZob{}} \PY{k}{return} \PY{n}{xs}\PY{o}{.}\PY{n+na}{get}\PY{o}{(}\PY{l+m+mi}{0}\PY{o}{)}\PY{o}{;} \PY{o}{\PYZcb{}}
\end{Verbatim}
}\noindent
or bounded wildcards
{\small\begin{Verbatim}[commandchars=\\\{\}]
\PY{n}{Number} \PY{n+nf}{firstNum}\PY{o}{(}\PY{n}{List}\PY{o}{\PYZlt{}}\PY{o}{?} \PY{k+kd}{extends} \PY{n}{Number}\PY{o}{\PYZgt{}} \PY{n}{xs}\PY{o}{)}
\PY{o}{\PYZob{}} \PY{k}{return} \PY{n}{xs}\PY{o}{.}\PY{n+na}{get}\PY{o}{(}\PY{l+m+mi}{0}\PY{o}{)}\PY{o}{;} \PY{o}{\PYZcb{}}
\end{Verbatim}
}\noindent
Both variants \emph{bound} the elements to be subtypes of ${\tt Number}$.
Now let us consider the converse situation,
  in which we want to insert an integer into the list.
In this case, Java lets us use only the variant with wildcards:
{\small\begin{Verbatim}[commandchars=\\\{\}]
\PY{k+kt}{void} \PY{n+nf}{addOne}\PY{o}{(}\PY{n}{List}\PY{o}{\PYZlt{}}\PY{o}{?} \PY{k+kd}{super} \PY{n}{Integer}\PY{o}{\PYZgt{}} \PY{n}{xs}\PY{o}{)}
\PY{o}{\PYZob{}} \PY{n}{xs}\PY{o}{.}\PY{n+na}{add}\PY{o}{(}\PY{l+m+mi}{1}\PY{o}{)}\PY{o}{;} \PY{o}{\PYZcb{}}
\end{Verbatim}
}\noindent
As with $\tt firstNum$, the method $\tt addOne$ can be used with both of
  ${\tt List}\langle{\tt Integer}\rangle$
  and
  ${\tt List}\langle{\tt Number}\rangle$.

Now let us consider the call to ${\tt firstNum}$,
  for a list of integers.
For the bounded generics version, the call is
\[
  {\tt this}.\langle{\tt Integer}\rangle{\tt firstNum}({\tt xs})
  \quad\text{or}\quad
  {\tt firstNum}({\tt xs})
\]
The latter is a simplified version,
  made possible by type inference and syntactic sugar.
So, let us focus on the former.
In it, we set the type variable $\tt x$ to ${\tt Integer}$,
  thus selecting which version of the generic method ${\tt firstNum}$ we are using.
On the other hand, for the bounded wildcard version, the call is
\[
  {\tt this}.{\tt firstNum}({\tt xs})
  \quad\text{or}\quad
  {\tt firstNum}({\tt xs})
\]
This code looks similar to the one for bounded generics,
  but it type checks for a different reason:
  ${\tt List}\langle{\tt Integer}\rangle$ is considered a subtype of
  ${\tt List}\langle\texttt{? extends Number}\rangle$,
  because ${\tt Integer}$ is a subtype of ${\tt Number}$.

Also, if $\tt B$~is a subtype of~$\tt A$, then
  ${\tt List}\langle\texttt{? extends B}\rangle$ is a subtype of
    ${\tt List}\langle\texttt{? extends A}\rangle$, and
  ${\tt List}\langle\texttt{? super A}\rangle$ is a subtype of
    ${\tt List}\langle\texttt{? super B}\rangle$.
Finally, ${\tt List}\langle\texttt{?}\rangle$ is a supertype of all lists.

Bounded wildcards are used liberally in the implementation of Java's standard library.
For a more interesting example, consider the following method from
  ${\sf java}.{\sf util}.{\sf Collections}$:
{\small\begin{Verbatim}[commandchars=\\\{\}]
\PY{k+kd}{static} \PY{o}{\PYZlt{}}\PY{n}{T}\PY{o}{\PYZgt{}} \PY{k+kt}{int} \PY{n}{binarySearch}\PY{o}{(}
  \PY{n}{List}\PY{o}{\PYZlt{}}\PY{o}{?} \PY{k+kd}{extends} \PY{n}{Comparable}\PY{o}{\PYZlt{}}\PY{o}{?} \PY{k+kd}{super} \PY{n}{T}\PY{o}{\PYZgt{}}\PY{o}{\PYZgt{}} \PY{n}{list}\PY{o}{,}
  \PY{n}{T} \PY{n}{key}\PY{o}{)} \PY{o}{\PYZob{}} \PY{c+cm}{/* ...*/} \PY{o}{\PYZcb{}}
\end{Verbatim}
}\noindent
`To search for a key of type $\tt T$ in a list,
  we must have a list whose elements are comparable to~$\tt T$.'
To express this constraint,
  we need both `{\tt ? extends}' and `{\tt ? super}'

Java's mechanism for deciding whether
  ${\tt List}\langle\alpha\rangle$
  is a subtype of
  ${\tt List}\langle\beta\rangle$
  is known as \df{use-site variance}
  because it involves getting a hint from inside $\alpha$~and~$\beta$:
  Do they mention `{\tt ? super}' or `{\tt ? extends}'?
The alternative, \df{declaration-site variance},
  is to take the hint from the declaration of $\tt List$,
  where the generic type can be annotated as being
    covariant, invariant, or contravariant.
We can simulate declaration site variance by use site variance as follows:
  if the declaration of~$\tt L$ annotates type~$\tt T$ to be covariant,
  then
    we erase the annotation and
    replace all uses
      ${\tt L}\langle\ldots,\texttt{T},\ldots\rangle$
      by ${\tt L}\langle\ldots,\texttt{? extends T},\ldots\rangle$;
  for contravariant annotations,
    we proceed similarly, but using `{\tt ? super}' instead of `{\tt ? extends}'.
In what follows, we only need the contravariant case,
  so it is the only one we formalize.

\section{Preliminaries}\label{sec:prelim} 

We formalize a fragment of Java's type system, following \cite{kennedy2007}.
\df{Types} are defined inductively:
  if $x$~is a type variable, then $x$~is a type;
  if $C$~is a class of arity $m\ge0$ and $t_1,\ldots,t_m$ are types,
    then $Ct_1\ldots t_m$ is a type.
Given a class~$C$, we write $\arity(C)$ for its arity;
given a type~$t$, we write $\vars(t)$ for the type variables occurring in~$t$.
A \df{substitution} $\sigma$ is a mapping
  from a finite set $\dom(\sigma)$ of type variables to types;
  we write $t\sigma$ for the result of applying $\sigma$ to~$t$.
A \df{class table} is a set of \df{inheritance rules} of the form
\begin{align*}
  Cx_1\ldots x_m \inh D t_1 \ldots t_n
\end{align*}
such that $\vars(t_j)\in\{x_1,\ldots,x_m\}$ for $j\in\{1,\ldots,n\}$,
  where $m=\arity(C)$ and $n=\arity(D)$.
A class table defines a binary relation $\inh$ between types:
  if $t_L \inh t_R$ is an inheritance rule,
  then $t_L\sigma \inh t_R\sigma$ for all substitutions~$\sigma$.
Note that we slightly abuse the notation~$\inh$,
  using it both as a binary relation on types
  and as a syntactic separator in inheritance rules.
When $\inh$~is used as a relation,
  we denote by $\inh^*$ its reflexive transitive closure.
(Note also that dropping the symbols `{\tt <,>}' that Java uses is unambiguous,
  as long as arities are known.)

Java forbids \df{multiple instantiation inheritance}:
  that is, in Java,
  if $t \inh^* C t_1\ldots t_m$ and $t \inh^* C t'_1 \ldots t'_m $,
  then $t_i=t'_i$ for $1 \le i \le m$.
In particular,
  if $Cx_1\ldots x_m \inh^* C t_1\ldots t_m$,
  then $x_i=t_i$ for $1\le i\le m$;
  that is, the class table is \df{acyclic}.

The subtyping relation $\subtype$ is defined inductively by the rule
\begin{align}
\dfrac
  {A t_1 \ldots t_m \inh^* C t'_1\ldots t'_n
    \qquad t''_1 \subtype t'_1 \quad\ldots\quad t''_n \subtype t'_n}
  {A t_1 \ldots t_m \subtype C t''_1 \ldots t''_n}
\label{proof-rule}
\end{align}
If $t \subtype t'$,
  we say that $t$~is a \df{subtype} of~$t'$,
  and $t'$~is a \df{supertype} of~$t$.
In rule~\eqref{proof-rule},
  types $t''_*$ occur in the supertype of the goal but in the subtype of assumptions.
In other words, we consider only the \df{contravariant} case.
It is known that Java's wildcards can encode declaration-site contravariance,
  and that subtyping is decidable in the absence of contravariance
  \cite{kennedy2007}.

\autoref{thm:main} is proved by a reduction
  from the halting problem of Turing machines.
A \df{Turing machine}~$\cT$
  is a tuple $(Q,q_{\sf I},q_{\sf H}, \Sigma,\delta)$,
  where $Q$~is a finite set of \df{states},
  $q_{\sf I}$~is the \df{initial state},
  $q_{\sf H}$~is the \df{halt state},
  $\Sigma$~is a finite \df{alphabet},
  and $\delta : Q\times\Sigma_{\bot} \to Q\times\Sigma\times\{\dL,\dS,\dR\}$
    is a \df{transition function}.
We require $\delta(q_{\sf H},a)=(q_{\sf H},a,\dS)$ for all~$a\in\Sigma$,
  and $\delta(q_{\sf H},\bot)=(q_{\sf H},a,\dS)$ for some~$a\in\Sigma$.
A \df{configuration} is a tuple $(q,\alpha,b,\gamma)$ of
  the \df{current state}~$q\in Q$,
  the \df{left part of the tape} $\alpha\in\Sigma^*$,
  the \df{current symbol} $b\in\Sigma_\bot$,
  and the \df{right part of the tape} $\gamma\in\Sigma^*$.
The \df{execution steps} of~$\cT$ are the following:
\begin{align*}
  (q,\alpha a,b,\gamma) &\to (q',\alpha,a,b'\gamma)
    &&\text{for $\delta(q,b)=(q',b',\dL)$}
\\
  (q,\alpha,b,\gamma) &\to (q',\alpha,b',\gamma)
    &&\text{for $\delta(q,b)=(q',b',\dS)$}
\\
  (q,\alpha,b,c\gamma) &\to (q',\alpha b',c,\gamma)
    &&\text{for $\delta(q,b)=(q',b',\dR)$}
\end{align*}
We also allow for execution steps that go outside the existing tape:
\begin{align*}
  (q,\epsilon,b,\gamma) &\to (q',\epsilon,\bot,b'\gamma)
    &&\text{for $\delta(q,b)=(q',b',\dL)$}
\\
  (q,\alpha,b,\epsilon) &\to (q',\alpha b',\bot,\epsilon)
    &&\text{for $\delta(q,b)=(q',b',\dR)$}
\end{align*}
(Here and throughout, $\epsilon$ stands for the empty string.)
A \df{run on input tape} $\alpha_{\sf I}$ is a sequence of execution steps
  starting from configuration $(q_{\sf I},\epsilon,\bot,\alpha_{\sf I})$.
If $\cT$~reaches $q_{\sf H}$ we say that $\cT$~\df{halts} on~$\alpha_{\sf I}$.

\begin{theorem}[\cite{turing}]
It is undecidable whether a Turing machine~$\cT$ halts on input~$\alpha_{\sf I}$.
\end{theorem}

\section{Subtyping Machines}\label{sec:subtype-machine} 

The fragment of Java's type system defined in the previous section
  does not seem to share much with Turing machines.
To clarify the connection,
  let us define a third formalism: subtyping machines.
The plan is to see how subtyping machines
  correspond to a fragment of a fragment of Java's type system,
  and at the same time can simulate Turing machines.
So far,
  we saw a fragment of Java's type system
  that included generic classes of arbitrary arity.
Subtyping machines can only handle the case
  in which all classes have arity~$1$,
  apart from one distinguished class $Z$ which has arity~$0$.
The configuration of a subtyping machine is a subtyping query
\[
  C_1 C_2 \ldots C_m Z \subtype D_1 D_2 \ldots D_n Z
\]
For reasons that will become clear later (\autoref{ex:<->}),
  we introduce two alternative notations for the \emph{same} configuration as above:
\begin{align*}
  Z C_m C_{m-1} \ldots C_1 &\< D_1 D_2 \ldots D_n Z \\
=\quad  Z D_n D_{n-1} \ldots D_1 &\> C_1 C_2 \ldots C_m Z
\end{align*}

Since we restrict our attention to arity~$\le1$,
  there are two ways in which proof rule~\eqref{proof-rule} can be applied:
\begin{gather*}
\frac%
  {C_1\ldots C_mZ \inh^* Z}%
  {\color{darkgreen} C_1\ldots C_m Z \subtype Z}
\\[1ex]
\frac%
  {C_1\ldots C_mZ \inh^* D_1E_2\ldots E_pZ
    \quad \color{darkblue} D_2\ldots D_nZ \subtype E_2\ldots E_pZ}%
  {\color{darkgreen} C_1\ldots C_mZ \subtype D_1\ldots D_nZ}
\end{gather*}
Correspondingly,
  we define two types of \df{execution steps} ${\cdot}\leadsto{\cdot}$
  for the subtyping machine:
\begin{align*}
  (\color{darkgreen}Z C_m\ldots C_1 \< Z\color{black})
  \;&\leadsto\; {\bullet}
\end{align*}
if $C_1\ldots C_mZ \inh^* Z$, and
\begin{align*}
  (\color{darkgreen}Z C_m \ldots C_1 \< D_1  \ldots D_n Z\color{black})
\;\leadsto\;
  (\color{darkblue}Z E_p \ldots E_2 \> D_2 \ldots D_n Z\color{black})
\end{align*}
if $C_1\ldots C_mZ \inh^* D_1E_2\ldots E_pZ$.
The special configuration~$\bullet$ is called the \df{halting} configuration.

Recall that $\inh^*$ is the reflexive transitive closure of a relation
  defined by the class table.
Thus, in particular,
  if the class table contains the inheritance rule
\[
  \color{darkgreen}C_1\color{black} x
    \inh \color{darkgreen}D_1 \color{darkblue}E_2 E_3 \ldots E_p \color{black} x
\]
then we can instantiate it with $x \defeq C_2 \ldots C_m Z$
  to enable the following execution step:
\begin{align}
\begin{aligned}
  & Z C_m \ldots C_2 \color{darkgreen} C_1
    \color{black} \< \color{darkgreen}D_1\color{black} D_2 \ldots D_n Z
\\\leadsto\quad
  & Z C_m \ldots C_2 \color{darkblue} E_p \ldots E_2
     \color{black} \> \color{black} D_2 \ldots D_n Z
\end{aligned}
  \label{eq:step1}
\end{align}
Also, because $\inh^*$~is reflexive, the following execution step is enabled:
\begin{align}
\begin{aligned}
  & Z C_m \ldots C_1 \color{darkgreen} N
    \color{black}\< \color{darkgreen} N \color{black} D_1 \ldots D_n Z
\\\leadsto\quad
  & Z C_m \ldots C_1 \> D_1 \ldots D_n Z
\end{aligned}
  \label{eq:step0}
\end{align}

The runs of subtyping machines correspond to (partial) proofs.
Runs that halt correspond to completed proofs.
Runs that get stuck correspond to failed proofs.
The subtyping machine may be nondeterministic,
  which corresponds to situations in which one may apply proof rule~\eqref{proof-rule}
    in several ways;
but,
  if multiple instantiation inheritance is forbidden, as it is in Java,
  then the subtyping machine is deterministic.
There is nothing deep about subtyping machines:
  they are introduced simply because the new notation will make it much easier
  to notice certain patterns.

\begin{proposition}
Consider a subtyping machine described by a given class table.
We have
  \[ C_1 \ldots C_m Z \subtype D_1\ldots D_n Z \]
if and only if there exists a halting run
  \[ (Z C_m\ldots C_1 \< D_1\ldots D_n Z) \;\leadsto^*\; {\bullet}\]
\end{proposition}

\begin{example}\label{ex:<->}
In preparation for the main reduction,
  let us look at one particular subtyping machine.
Consider the following class table:
\begin{align}
  Q^{\sf L}x &\inh LNQ^{\sf L}LNx
&
  Q^{\sf R}x &\inh LNQ^{\sf R}LNx   \label{ex:ql}
\\
  Q^{\sf L}x &\inh EQ^{\sf LR}Nx
&
  Q^{\sf R}x &\inh EQ^{\sf RL}Nx    \label{ex:qe}
\\
  Ex &\inh Q^{\sf LR}NQ^{\sf R}EEx
&
  Ex &\inh Q^{\sf RL}NQ^{\sf L}EEx \label{ex:eq}
\end{align}
and the query $Q^{\sf R}EEZ \subtype LNLNLNEEZ$.
Then, the subtyping machine runs as follows:
\begin{align*}
  &ZEEQ^{\sf R} \< LNLNLNEEZ
\\\leadsto^2\quad
  &ZEENLQ^{\sf R} \< LNLNEEZ
  &&\text{by \eqref{ex:ql}+\eqref{eq:step1}, then \eqref{eq:step0}}
\\\leadsto^2\quad
  &ZEENLNLQ^{\sf R} \< LNEEZ
  &&\text{by \eqref{ex:ql}+\eqref{eq:step1}, then \eqref{eq:step0}}
\\\leadsto^2\quad
  &ZEENLNLNLQ^{\sf R} \< EEZ
  &&\text{by \eqref{ex:ql}+\eqref{eq:step1}, then \eqref{eq:step0}}
\\\leadsto^{\phantom{1}}\quad
  &ZEENLNLNLNQ^{\sf RL} \> EZ
  &&\text{by \eqref{ex:qe}+\eqref{eq:step1}}
\\\leadsto^{\phantom{1}}\quad
  &ZEENLNLNLN \< NQ^{\sf L}EEZ
  &&\text{by \eqref{ex:eq}+\eqref{eq:step1}}
\\\leadsto^{\phantom{1}}\quad
  &ZEENLNLNL \> Q^{\sf L}EEZ
  &&\text{by \eqref{eq:step0}}
\\\leadsto^2\quad
  &ZEENLNL \> Q^{\sf L}LNEEZ
  &&\text{by \eqref{ex:ql}+\eqref{eq:step1}, then \eqref{eq:step0}}
\\\leadsto^2\quad
  &ZEENL \> Q^{\sf L}LNLNEEZ
  &&\text{by \eqref{ex:ql}+\eqref{eq:step1}, then \eqref{eq:step0}}
\\\leadsto^2\quad
  &ZEE \> Q^{\sf L}LNLNLNEEZ
  &&\text{by \eqref{ex:ql}+\eqref{eq:step1}, then \eqref{eq:step0}}
\end{align*}
Some lines group together two execution steps
  so that the overall pattern is clearer:
  we have a head traveling back-and-forth over a tape!
The Java code corresponding to \eqref{ex:ql}, \eqref{ex:qe}, \eqref{ex:eq}
  appears in \autoref{fig:java<->}.
\end{example}

\begin{figure}
{\footnotesize\include{back-and-forth}}
\caption{Java code for \autoref{ex:<->}}
\label{fig:java<->}
\end{figure}

\paragraph{Well-formedness.}
From \cite{kennedy2007}, we know that the subtyping relation is transitive
  if and only if
  certain well-formedness conditions hold.
Instead of stating these conditions in full generality,
  let us do it only for our special case,
  in which arities are~$\le 1$ and everything is contravariant.
Under these conditions,
  \df{well-formedness} requires that,
  if an inheritance rule has the form $Ax \inh D_1\ldots D_n x$,
  then $n$~must be odd.

\section{Main Result}\label{sec:result} 

Given a Turing machine~$\cT$ and an input tape~$\alpha_{\sf I}$,
  we will construct types $t_1$, $t_2$ and a class table
  such that $t_1 \subtype t_2$ if and only if $\cT$ halts on~$\alpha_{\sf I}$.
For each state $q_s \in Q$, we have six classes:
  $Q_s^{\sf wL}$, $Q_s^{\sf wR}$,
  $Q_s^{\sf L}$, $Q_s^{\sf R}$,
  $Q_s^{\sf LR}$ and $Q_s^{\sf RL}$;
  for each letter $a \in \Sigma\cup\{\#\}$, we have a class~$L_a$.
Here, $\#$ is a fresh letter.
We also make use of four auxiliary classes:
  $N$, $E$, $M^{\sf L}$, and~$M^{\sf R}$.

\begin{remark}
Informally, these classes will be used as follows.
A class $Q_s^*$ indicates that we are simulating the Turing state~$s$.
A class $Q_*^{\sf w{*}}$ indicates that we are waiting:
  the head of the subtyping machine
  is not in the same position as the head of the simulated Turing machine;
conversely, if {\sf w} is missing, then the head of the subtyping machine
  is in the same position as the head of the simulated Turing machine.
The superscripts {\sf L}/{\sf R} indicate that the head of the subtyping machine
  is moving {\bf l}eft\slash {\bf r}ight.
The class $Q_*^{\sf LR}$ indicates that the subtyping machine used to move towards
  {\bf l}eft but is now in the process of turning {\sf r}ight.
The classes $M^{\sf L}$~and~$M^{\sf R}$ mark the position of the Turing machine head,
  which can be on the {\bf l}eft of the marker or on the {\bf r}ight of the marker.
The special letter $\#$ marks the two endpoints of the tape,
  and is used for extending the tape of the subtyping machine,
  if more space is needed.
The class $E$ is an end of tape marker, on the outside of~$\#$,
  which helps the subtyping machine head turn around.
Roughly, the class $N$ is a trick that lets us have covariance if we want it,
  without putting it in the formalism.
\end{remark}

The subtyping machine will have configurations
  of one of the forms in \autoref{fig:configs}.
These configurations obey several simple invariants.
If we read off only the symbols of the form $L_*$, we~get
\[
  L_{a_i} L_{a_{i+1}} \ldots L_{a_{k-1}} L_{a_k}
\]
This will represent the tape of a Turing machine,
  with content $a_{i+1}\ldots a_{k-1}$.
We use $L_{a_i}$~and~$L_{a_k}$ as markers of the two endpoints;
  in fact we impose the invariant that $i<k$ and $a_i=a_k=\#$.

\begin{figure*}
\begin{align*}
& ZEE\, NL_{a_i} \ldots NL_{a_l} Q_s^{\sf wR}
  \blacktriangleleft
  L_{a_{l+1}}N \ldots L_{a_{j-1}}N\, M^{\sf R}N\, L_{a_j}N \ldots L_{a_k}N\, EEZ
  &&\text{for $l+1\le j\le k$}
\\
& ZEE\, NL_{a_i} \ldots NL_{a_l} Q_s^{\sf wR}
  \blacktriangleleft
  L_{a_{l+1}}N \ldots L_{a_j}N\, M^{\sf L}N\, L_{a_{j+1}}N \ldots L_{a_k}N\, EEZ
  &&\text{for $l+1 \le j \le k$}
\\
& ZEE\, NL_{a_i} \ldots NL_{a_{j-1}}\, NM^{\sf R}\, NL_{a_j}  \ldots NL_{a_l}
  Q_s^{\sf wR}\blacktriangleleft
  L_{a_{l+1}}N \ldots L_{a_k}N\, EEZ
  &&\text{for $i\le j \le l$}
\\
& ZEE\, NL_{a_i} \ldots NL_{a_j}\, NM^{\sf L}\, NL_{a_{j+1}}  \ldots NL_{a_l}
  Q_s^{\sf wR}\blacktriangleleft
  L_{a_{l+1}}N \ldots L_{a_k}N\, EEZ
  &&\text{for $i\le j\le l$}
\\
& ZEE\, NL_{a_i} \ldots NL_{a_l}
  Q_s^{\sf R}\blacktriangleleft
  L_{a_{l+1}}N \ldots L_{a_k}N\, EEZ
  &&\text{for $j=l+1$}
\\[1ex]
& ZEE\, NL_{a_i} \ldots NL_{a_l}
  \blacktriangleright Q_s^{\sf wL}
  L_{a_{l+1}}N \ldots L_{a_{j-1}}N\, M^{\sf R}N\, L_{a_j}N \ldots L_{a_k}N\, EEZ
  &&\text{for $l+1 \le j \le k$}
\\
& ZEE\, NL_{a_i} \ldots NL_{a_l}
  \blacktriangleright Q_s^{\sf wL}
  L_{a_{l+1}}N \ldots L_{a_j}N\, M^{\sf L}N\, L_{a_{j+1}}N \ldots L_{a_k}N\, EEZ
  &&\text{for $l+1 \le j \le k$}
\\
& ZEE\, NL_{a_i} \ldots NL_{a_{j-1}}\, NM^{\sf R}\, NL_{a_j}  \ldots NL_{a_l}
  \blacktriangleright Q_s^{\sf wL}
  L_{a_{l+1}}N \ldots L_{a_k}N\, EEZ
  &&\text{for $i\le j\le l$}
\\
& ZEE\, NL_{a_i} \ldots NL_{a_j}\, NM^{\sf L}\, NL_{a_{j+1}}  \ldots NL_{a_l}
  \blacktriangleright Q_s^{\sf wL}
  L_{a_{l+1}}N \ldots L_{a_k}N\, EEZ
  &&\text{for $i\le j\le l$}
\\
& ZEE\, NL_{a_i} \ldots NL_{a_l}
  \blacktriangleright Q_s^{\sf L}
  L_{a_{l+1}}N \ldots L_{a_k}N\, EEZ
  &&\text{for $j=l$}
\end{align*}
\caption{
  The form of configurations of the subtyping machine
    being constructed  in \autoref{sec:result}.
  All these subtyping machine configurations simulate the same
    Turing machine configuration;
    namely, $(q_s,a_{i+1}\ldots a_{j-1}, {\color{darkblue}a_j}, a_{j+1}\ldots a_{k-1})$.
  By convention, $a_i=a_k={\#}$.
  The head of the subtyping machine is between $l$ and $l+1$;
    in all cases, $i-1\le l\le k$.
  The head of the simulated Turing machine is at position~$j$.
}
\label{fig:configs}
\end{figure*}

If, in addition to $L_*$,
  we also read off $M^{\sf L}$, $M^{\sf R}$, $Q_s^{\sf L}$ and $Q_s^{\sf R}$,
  then we get one of
\begin{align*}
  L_{a_i} \ldots L_{a_{j-1}}&\color{darkblue}L_{a_j}\color{black}
    M^{\sf L}
  L_{a_{j+1}} \ldots L_{a_k}
\\
  L_{a_i} \ldots L_{a_{j-1}}&\color{darkblue}L_{a_j}\color{black}
    Q_s^{\sf L}
  L_{a_{j+1}} \ldots L_{a_k}
\\
  L_{a_i} \ldots L_{a_{j-1}}
    M^{\sf R}
  &\color{darkblue}L_{a_j}\color{black} L_{a_{j+1}}\ldots L_{a_k}
\\
  L_{a_i} \ldots L_{a_{j-1}}
    Q_s^{\sf R}
  &\color{darkblue}L_{a_j}\color{black} L_{a_{j+1}}\ldots L_{a_k}
\end{align*}
This is how we mark the current symbol~$a_j$ of the simulated Turing machine:
  $M^{\sf L}$~and~$Q_s^{\sf L}$ indicate that the head is on the letter to the left;
  $M^{\sf R}$~and~$Q_s^{\sf R}$ indicate that the head is on the letter to the right.
The markers $M^{\sf L}$~and~$M^{\sf R}$
  always have at least one letter on both their sides;
the marker $Q_s^{\sf L}$ must have a letter on its left
  (but not necessarily on it right);
the marker $Q_s^{\sf R}$ must have a letter on its right
  (but not necessarily on it left).
These conditions account for the constraints on $i,j,k,l$ in \autoref{fig:configs}.

If we read off only the symbols of the form
  $Q_s^{\sf L}$, $Q_s^{\sf R}$, $Q_s^{\sf wL}$ and $Q_s^{\sf wR}$,
  then we see that each configuration in \autoref{fig:configs}
    contains exactly one such symbol:
It tells us that the Turing machine we simulate is in state~$q_s$.

Putting all this together,
  we see that all subtyping machine configurations in \autoref{fig:configs}
  simulate the same configuration of the Turing machine,
  namely $(q_s,a_{i+1}\ldots a_{j-1},a_j,a_{j+1}\ldots a_{k-1})$.

\smallskip

\begin{figure*}
\def\CR#1{{\color{darkred}#1}}
\def\CG#1{{\color{darkgreen}#1}}
\def\CB#1{{\color{darkblue}#1}}
\begin{align*}
\CR{Q_s^{\sf wL}}x &\inh M^{\sf L} N \CR{Q_s^{\sf L}} x
  &&\text{for $\CR{q_s} \in Q$}
&
\CR{Q_s^{\sf L}} x &\inh \CR{L_a N} \CG{Q_{s'}^{\sf wL}} \CB{M^{\sf L} N} \CG{L_b N} x
  &&\text{for $\delta(\CR{q_s},\CR{a})=(\CG{q_{s'}},\CG{b},\CB{\dL})$}
\\
\CR{Q_s^{\sf wL}}x &\inh M^\dR N \CR{Q_s^{\sf wL}} M^\dR N x
  &&\text{for $\CR{q_s} \in Q$}
&
\CR{Q_s^{\sf L}} x &\inh \CR{L_a N} \CG{Q_{s'}^{\sf wL}} \CB{M^{\sf R} N} \CG{L_b N} x
  &&\text{for $\delta(\CR{q_s},\CR{a})=(\CG{q_{s'}},\CG{b},\CB{\dS})$}
\\
\CR{Q_s^{\sf wL}}x &\inh \CB{L_a N} \CR{Q_s^{\sf wL}} \CB{L_a N} x
  &&\text{for $\CR{q_s} \in Q$ and $\CB{a} \in \Sigma\cup\{\#\}$}
&
\CR{Q_s^{\sf L}} x &\inh \CR{L_a N} \CG{Q_{s'}^{\sf wL}} \CG{L_b N} \CB{M^{\sf R} N} x
  &&\text{for $\delta(\CR{q_s},\CR{a})=(\CG{q_{s'}},\CG{b},\CB{\dR})$}
\\
\CR{Q_s^{\sf wL}} x &\inh E \CR{Q_s^{\sf LR}} N x
  &&\text{for $\CR{q_s} \in Q\setminus\{q_{\sf H}\}$}
&
\CR{Q_s^{\sf L}} x
  &\inh \CR{L_\# N} \CG{Q_{s'}^{\sf wL}} \CG{L_\# N} \CB{M^{\sf L} N} \CG{L_b N} x
  &&\text{for $\delta(\CR{q_s},\CR{\bot})=(\CG{q_{s'}},\CG{b},\CB{\dL})$}
\\
Q_{\sf H}^{\sf wL} x &\inh E E Z &&
&
\CR{Q_s^{\sf L}} x
  &\inh \CR{L_\# N} \CG{Q_{s'}^{\sf wL}} \CG{L_\# N} \CB{M^{\sf R} N} \CG{L_b N} x
  &&\text{for $\delta(\CR{q_s},\CR{\bot})=(\CG{q_{s'}},\CG{b},\CB{\dS})$}
\\
E x &\inh \CR{Q_s^{\sf LR}} N \CR{Q_s^{\sf wR}} EEx
  &&\text{for $\CR{q_s} \in Q$}
&
\CR{Q_s^{\sf L}} x
  &\inh \CR{L_\# N} \CG{Q_{s'}^{\sf wL}} \CG{L_\# N L_b N} \CB{M^{\sf R} N} x
  &&\text{for $\delta(\CR{q_s},\CR{\bot})=(\CG{q_{s'}},\CG{b},\CB{\dR})$}
\end{align*}
\caption{
  Class table used to simulate a Turing machine
    $\cT=(Q,q_{\sf I}, q_{\sf H}, \Sigma, \delta)$.
  There are twelve more inheritance rules,
    obtained by swapping $\dL\leftrightarrow\dR$ in the rules above.
}
\label{fig:reduction}
\end{figure*}

Now let us move to the class table,
  which describes how the subtyping machine runs.
The inheritance rules are given in \autoref{fig:reduction}.
We want to start the subtyping machine from the configuration
\[
  ZEE\,
    N L_\#\, NM^\dL\, NL_{a_1} \ldots NL_{a_m}\, NL_\#\,
    Q_{\sf I}^{\sf wR} \blacktriangleleft
  EEZ
\]
where $a_1 \ldots a_m$ is the content of the initial tape~$\alpha_{\sf I}$.
Equivalently, we can say that we want to ask the subtyping query
\[
  Q_{\sf I}^{\sf wR}\,
  L_\#N \,
  L_{a_m} N \ldots L_{a_1}N \,
  M^\dL N \,
  L_\# N \,
  EEZ
    \subtype EEZ
\]

It is trivial to check that the inheritance rules in \autoref{fig:reduction}
  are well-formed.
Now let us check whether they use multiple instantiation inheritance.
For this, we build an inheritance graph,
  which has an arc $A\to B$ for each inheritance rule of the form $Ax \inh Bt$.
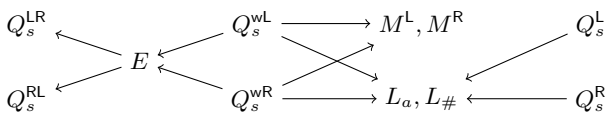
\begin{center}
\begin{tikzpicture}[yscale=0.5,xscale=1.5]
  \node (qilr) at (0,2) {$Q_s^{\sf LR}$};
  \node (qirl) at (0,0) {$Q_s^{\sf RL}$};
  \node (e) at (1,1) {$E$};
  \node (qiwl) at (2,2) {$Q_s^{\sf wL}$};
  \node (qiwr) at (2,0) {$Q_s^{\sf wR}$};
  \node (mlr) at (3.5,2) {$M^{\sf L}, M^{\sf R}$};
  \node (la) at (3.5,0) {$L_a,L_{\#}$};
  \node (qil) at (5,2) {$Q_s^{\sf L}$};
  \node (qir) at (5,0) {$Q_s^{\sf R}$};
  \draw[->] (e) -- (qilr);
  \draw[->] (e) -- (qirl);
  \draw[->] (qiwl) -- (e);
  \draw[->] (qiwl) -- (mlr);
  \draw[->] (qiwl) -- (la);
  \draw[->] (qiwr) -- (e);
  \draw[->] (qiwr) -- (mlr);
  \draw[->] (qiwr) -- (la);
  \draw[->] (qil) -- (la);
  \draw[->] (qir) -- (la);
\end{tikzpicture}
\end{center}
Recall that the Turing machine we simulate is deterministic;
  in particular, for each $(q_s,a)$ there is a single arc $Q_s^{\sf L}\to L_a$.
We see that all walks (which are possibly self-intersecting paths)
  are uniquely identified by their endpoints.
Thus, multiple instantiations are not possible.

The class table is large, but it is based on a few simple ideas.
In the middle of the tape, all evolutions follow the following pattern,
  which is a simple generalization of what we saw in \autoref{ex:<->}:
\begin{align*}
& \alpha Q \blacktriangleleft S N \beta
\\\leadsto\quad
& \alpha N S_1 NS_2 Q'N \blacktriangleright N \beta
  &&\text{by \eqref{eq:step1} and a rule \eqref{eq:mid-rule}}
\\\leadsto\quad
& \alpha N S_1 NS_2 Q' \blacktriangleleft \beta
  &&\text{by \eqref{eq:step0}}
\end{align*}
Here,
  (i)~$\alpha$~and~$\beta$ are some strings of symbols,\;
  (ii)~$S,S_1,S_2$ are letters~$L_*$ or markers~$M^*$, \; and
  (iii)~the class table contains an inheritance rule of the form
\begin{align}
Qx \inh SN\,Q'\, S_2N\, S_1N\, x
  \label{eq:mid-rule}
\end{align}
If we ignore the padding with $N$ temporarily,
  we see that the effect is the following rewriting:
\[
  \color{darkred}Q \blacktriangleleft \color{darkgreen}S\color{black}
  \quad\Longrightarrow\quad
  \color{darkgreen}S_1 S_2 \color{darkred}Q' \blacktriangleleft {}
\]
The `state' changes from~$Q$ to~$Q'$,
  and the `symbol' $S$ is replaced with $S_1S_2$.
Also, the `head' $\blacktriangleleft$ passes over the symbol~$S$ that it processed.
In fact, the head of the subtyping machine will keep moving back-and-forth,
  just as in \autoref{ex:<->}.

The reversal of the `head' at the end of the `tape'
  is performed exactly as in \autoref{ex:<->}.
During reversal, configurations contain $Q_s^{\sf RL}$~or~$Q_s^{\sf LR}$,
  thus stepping outside of the set listed in \autoref{fig:configs},
  but only temporarily.

At this point we have a head that traverses the tape back-and-forth,
  carries with it a state,
  and can replace a symbol on the tape by several symbols.
It is now a simple exercise to figure out
  how to use markers $M^{\sf L}$~and~$M^{\sf R}$ to track
  the head of the Turing machine we simulate.
It is also a simple exercise to figure out how to use the markers $L_\#$
  at the endpoints of the tape in order to handle the case
  in which the tape needs to be extended.
Finally, it is an easy exercise to figure out how to make
  the subtyping machine halt (soon) after it reaches a configuration
  that simulates the halting state~$q_{\sf H}$ of the Turing machine.
In fact, all these exercises are solved in \autoref{fig:reduction}.

We have proved the main result:
\thmmain*
\noindent
The construction described in the previous proof is implemented
  \cite{website},
  so that the reader can test the proof by trying it on example Turing machines.
The examples include
  counting in binary,
  computing the Ackermann function,
  and checking whether ZF~is consistent.

\section{Application: Fluent Interfaces} \label{sec:fluent} 

The reduction from Turing machines (\autoref{sec:result})
  lets us conclude that typechecking is undecidable.
But, it is not immediately clear
  whether we can harness the reduction to do something useful.
In this section,
  we see that the reduction can be used to implement a parser generator
  for fluent interfaces.
The parser generator is efficient in theory (\autoref{thm:fluent}),
  but not in practice.
Before \autoref{thm:fluent},
  we had no reason to believe
    that a practically efficient parser generator for fluent interfaces is possible;
after \autoref{thm:fluent},
  we have no reason to believe that it is impossible.

\smallskip

Consider a Java library.
We say that it has a \df{fluent interface}
  when it encourages its users to chain method calls:
  $f().g().h()$.
Programs written in the fluent style are more readable
  than programs written in the traditional, sequential style:
  $f(); g(); h()$.
Why?
In the fluent style,
  the order in which methods are called is constrained by the type checker.
If the library is well designed,
  its types will constrain the chains of method calls
  such that they are easy to read and understand.
In the traditional style,
  we cannot use the type checker to enforce readability.

\cite{det-cfg-java} propose a principled method for designing fluent interfaces.
We start by describing the allowed chains of method calls using a context free grammar:
  method names are tokens,
  and allowed chains of method calls correspond to words in a context free language.
Their main result is the following:

\begin{theorem}[\cite{det-cfg-java}]
Given is a \textcolor{darkred}{deterministic} context free grammar~$G$
  that describes a language $\cL\subseteq\Sigma^*$
  over an alphabet~$\Sigma$ of method names.
We can construct
  Java class definitions,
  a type~$T$,
  and expressions ${\it Start}$, ${\it Stop}$
such that the code
\[ T\; \ell \quad=\quad
  {\it Start}.f^{(1)}().f^{(2)}()\ldots f^{(m)}().{\it Stop} \]
type checks if and only if $f^{(1)}f^{(2)}\ldots f^{(m)} \in \cL$.
\end{theorem}

Their construction relies on a conversion
  from deterministic pushdown automata
  to realtime jump-deterministic pushdown automata
  \cite{courcelle}.
Instead, we can use \autoref{thm:main} to show a similar but stronger result:
  (a)~the grammar
    will not be required to be deterministic,
  (b)~the generated code
    will be guaranteed to be polynomial in the size of the grammar, and
  (c)~the generated code can be type checked in polynomial time.
Point~(a) means that
  instead of a parser generator for ${\rm LR}(k)$ grammars,
  we give a parser generator for all context free grammars.
But, the key improvements are points (b)~and~(c),
  which mean that the parser is (theoretically) efficient.

\thmfluent*

\paragraph{Proof outline.}
From the grammar~$G$, we generate a CYK parser,
  implemented in a simple imperative language.
Then, we compile the parser into a Turing machine.
Then, we compile the Turing machine into Java class definitions,
  as in \autoref{fig:reduction}.
Finally, we construct the expression
\begin{align}
  {\it Start}.f^{(1)}()\ldots f^{(m)}().{\it Stop}
  \label{eq:java-word}
\end{align}
so that it has a type~$S$
  and so that checking whether $S \subtype T$
    is equivalent to running the parser on the word $f^{(1)}\ldots f^{(m)}$.
\qed

\medskip
Once we know \autoref{thm:main},
  it is not difficult to see that the transformations mentioned above are possible,
  even if working out all details is tedious.
In what follows,
  we unravel enough details to clarify that the construction is polynomial.
This description will also give a high level view of the accompanying
  proof of concept implementation~\cite{website}.

\subsection{Builders}\label{sec:builders} 

We start with the last part of the proof:
  how to choose the expressions ${\it Start}$ and ${\it Stop}$,
  which occur in~\eqref{eq:java-word}.
We want to initialize the subtyping machine in the following configuration:
\[
  \overbrace{
  \underbrace{\strut ZEENL_\# N M^\dL}_{\strut {\it Start}}
  \underbrace{\strut N L_{f^{(1)}}}_{\strut f^{(1)}()}
    \ldots
  \underbrace{\strut NL_{f^{(m)}}}_{\strut f^{(m)}()}
  \underbrace{\strut NL_\# Q_{\sf I}^{\sf wR}}_{\strut {\it Stop}}
  }^{\strut S}
  \<
  \overbrace{\strut EEZ}^{\strut T}
\]
Over braces, we have shorthand names for the types:
\begin{align*}
  S &\defeq
    Q_{\sf I}^{\sf wR}
    L_\# N
      L_{f^{(m)}} N \ldots L_{f^{(1)}} N
      M^\dL N
    L_\# N
    T \\
  T &\defeq EEZ
\end{align*}
Under braces,
  we see how various parts of the expression~\eqref{eq:java-word} build the type~$S$.
More precisely,
  we make use of an abstract class~$B$
  that is implemented as in \autoref{fig:builder}.
Given such a class~$B$, we can easily build the required type~$S$; for example,
\begin{Verbatim}[commandchars=\\\{\}]
  \PY{n}{E}\PY{o}{\PYZlt{}}\PY{o}{?}\PY{k+kd}{super} \PY{n}{E}\PY{o}{\PYZlt{}}\PY{o}{?}\PY{k+kd}{super} \PY{n}{Z}\PY{o}{\PYZgt{}}\PY{o}{\PYZgt{}} \PY{n}{l} \PY{o}{=}
    \PY{n}{B}\PY{o}{.}\PY{n+na}{start}\PY{o}{.}\PY{n+na}{a}\PY{o}{(}\PY{o}{)}\PY{o}{.}\PY{n+na}{a}\PY{o}{(}\PY{o}{)}\PY{o}{.}\PY{n+na}{b}\PY{o}{(}\PY{o}{)}\PY{o}{.}\PY{n+na}{c}\PY{o}{(}\PY{o}{)}\PY{o}{.}\PY{n+na}{stop}\PY{o}{(}\PY{o}{)}\PY{o}{;}
\end{Verbatim}

would check if the word ${\it aabc}$ is in the language being considered.

\begin{figure}
\include{builder}
\caption{
  Builder class, for obtaining a Turing tape out of a chain of method calls.
  In this example, the alphabet is $\Sigma=\{a,b,c\}$.
  In the general case, there is one method for each letter.
}
\label{fig:builder}
\end{figure}

\subsection{Background: Grammars and Parsers} \label{sec:cyk} 

We continue with the first part of the proof:
  how to go from a context free grammar to a parser.
The content of this section is standard,
  but it is included for completeness and ease of reference.

Given a finite set~$\Sigma$,
  we aim to specify subsets of~$\Sigma^*$.
We call $\Sigma$ an \df{alphabet},
  we call its elements \df{letters} or \df{terminal symbols} or \df{terminals},
  and we call subsets of~$\Sigma^*$~\df{languages}.
To define languages,
  we introduce another finite set~$\Gamma$,
  whose elements we call \df{nonterminal symbols} or \df{nonterminals}.
We denote terminals by lowercase Latin letters ($a$,~$b$, $c$,~\dots),
  and nonterminals by uppercase Latin letters ($A$,~$B$, $C$,~\dots).
We denote strings of symbols, from $\Sigma\cup\Gamma$,
  by lowercase Greek letters ($\alpha$,~$\beta$, $\gamma$, \dots),
  with $\epsilon$ reserved for the empty string.
We write $|\alpha|$ for the length of~$\alpha$.
We denote symbols (terminal or nonterminal) by $\theta_1, \ldots, \theta_n$.
A \df{production} is a relation $A\to\theta_1\ldots\theta_n$ for $n\ge0$,
  where $A$~is a nonterminal and $\theta_1\ldots\theta_n$ is a string of symbols.
A (context free) \df{grammar}~$G$ is a set of productions
  together with a distinguished \df{start nonterminal}~$S$.

Given a grammar, we can use its productions to inductively define \df{parse trees}:
(a)~if $a$ is a nonterminal, then $a$~is a parse tree of~$a$; and
(b)~if $A\to\theta_1\ldots\theta_n$ is a production,
  and $T_k$~is a parse tree of the symbol~$\theta_k$ for all $k\in\{1,\ldots,n\}$,
  then $(A,[T_1,\ldots,T_n])$ is a parse tree of~$A$.
If the leaves of a parse tree~$T$ read from left to write give $\alpha\in\Sigma^*$,
  we say that \df{$T$~is a parse tree of~$\alpha$}.
The \df{language}~$\cL$ defined by~$G$ is the set of strings $\alpha\in\Sigma^*$
  that have some parse tree among the parse trees of the start nonterminal~$S$.

Given string~$\alpha$ and a grammar~$G$ the membership problem asks
  whether $\alpha\in\cL$, where $\cL$~is the language defined by~$G$.
The membership problem can be solved by a straightforward recursion,
  known as the CYK algorithm.
We define a relation $\cyk$ inductively:
(a)~if $a$~is a nonterminal, then $\cyk(a,a)$ holds; and
(b)~if $A\to\theta_1\ldots\theta_n$ is a production,
  $\alpha$~is the concatenation $\alpha_1\ldots\alpha_n$,
  and $\cyk(\alpha_k,\theta_k)$ holds for all $k\in\{1,\ldots,n\}$,
  then $\cyk(\alpha,A)$ holds.
We have $\alpha\in\cL$ if and only if $\cyk(\alpha,S)$ holds,
  where $S$~is the start nonterminal.

One could compute $\cyk$ by simply using the definition from above
  and a Datalog engine that knows about strings.
Alternatively,
  one could compute $\cyk$ using a memoized recursive function,
  but then some care is needed to avoid nontermination.
Because the definition of $\cyk$ is monotone, the fix is easy:
If $\cyk(\alpha,A)$~is called while computing $\cyk(\alpha,A)$,
  then return ${\bf false}$ immediately.
The intuition is that, in an inductive definition (as opposed to a coinductive one),
  $\cyk(\alpha,A)$ cannot justify itself.

What is the complexity of the CYK algorithm?
The first argument of $\cyk$ is a substring of~$\alpha$,
  so there are $O\bigl(|\alpha|^2\bigr)$ possibilities.
If the second argument is a terminal, then $\cyk$~does only a constant amount of work;
  thus, the work for all terminals is~$O\bigl(|\alpha|^2\cdot|\Sigma|\bigr)$.
If the second argument is a nonterminal~$A$,
  then the algorithm iterates over all productions of the form
  $A\to\theta_1\ldots\theta_n$.
For each such production,
  the string $\alpha$ is split into $n$~substrings $\alpha_1,\ldots,\alpha_n$.
There are $\binom{|\alpha|+n-1}{n-1}$ ways to do so, which is exponential in~$n$.
For example, if $|\alpha|=n$, then the number of ways is $\Theta(4^n/\sqrt{n})$.

Luckily, there is a simple fix:
  We can transform the grammar into another one
  which defines the same language but has short productions.
We can do so by repeatedly replacing productions of the form
    $A\to\theta_1\ldots\theta_i\theta_{i+1}\ldots\theta_n$
  by
    $A \to \theta_1\ldots\theta_i A'$ and
    $A' \to \theta_{i+1}\ldots\theta_n$,
  where $A'$~is a fresh nonterminal.
This will increase the size of the grammar.
Let us define the size formally as
\[
  |G| \defeq |\Sigma| + \sum_{\mathclap{A\to\theta_1\ldots\theta_n}} (n+1)
\]
which is the size of the alphabet~$\Sigma$ plus the number of symbols necessary
  to write down the productions of~$G$.
Each replacement increases $|G|$ by~$2$,
  because of the two occurrences of the fresh nonterminal~$A'$.
If we take care to apply replacements only when both $\theta_i$~and~$\theta_{i+1}$
  come from the original grammar
  (as opposed to being fresh nonterminals introduced by previous replacements),
  then we will do at most $|G|$ replacements.
So, overall, the size~$|G|$ of the grammar increases only linearly.

Now suppose that, using the previous transformation,
  we limit the maximum size of productions to $n=2$.
Then, for one fixed production,
  there are $|\alpha|+1$ possible ways to split $\alpha$ into substrings.
Thus, the total runtime is $O\bigl(|\alpha|^3\cdot|G|\bigr)$.

The CYK algorithm is usually implemented
  not in the memoized version discussed so far
  but in a dynamic programming version.
In dynamic programming,
  we must be explicit about the order
  in which we evaluate $\cyk(\alpha,A)$.
It is clear we should should look at substrings from shortest to longest.
But, it is not clear in which order we should look at nonterminals.
The standard solution to this conundrum is to preprocess the grammar even more.
We compute the set $\{\,A\mid\cyk(\epsilon,A)\,\}$ of \df{nullable} nonterminals,
  perhaps by using our memoized version of~$\cyk$.
Then, for each production
  $A\to\theta_1\ldots\theta_{i-1}B\theta_{i+1}\ldots\theta_n$
  where $B$~is a nullable nonterminal,
  we introduce a new production
  $A\to\theta_1\ldots\theta_{i-1}\theta_{i+1}\ldots\theta_n$,
  and repeat this until a fixed-point is reached.
Finally, we remove all productions of the form $A\to\epsilon$.
In the resulting grammar,
  $\cyk(\epsilon,A)$ is ${\bf false}$ for all nonterminals~$A$,
  but $\cyk$~remains otherwise unchanged.
The grammar increases in size,
  but only linearly if the productions had size bounded by a constant.
The main advantage of the new grammar is that it is now safe split $\alpha$
  only in \emph{nonempty} substrings $\alpha_1,\ldots,\alpha_n$.
This is not sufficient to ensure a unique order, dependant just on~$|G|$,
  which means we still have to perform a fixed point computation to handle
  the productions of the form $A\to\theta$.

\begin{figure}
\begin{alg}
\^  $\proc{member}(\alpha)$
\0  if $\alpha=\epsilon$ return $\nullable(S)$
\0  initialize table $T$ with ~false~ everywhere
\0  for $i\in\{0,1,\ldots,|\alpha|-1\}$
\0    \comment for each terminal $a$,
\1    if $\alpha[i]=a$ then $T[i,i+1,a]:={\bf true}$
\0  for $k\in\{2,3,\ldots,|\alpha|\}$
\1    for $i\in\{0,1,\ldots,|\alpha|-k\}$
\2      for $j\in\{i+1,i+2,\ldots,i+k-1\}$
\2        \comment for each production $A\to\theta_1\theta_2$,
\3        if $T[i,j,\theta_1] \land T[j,i+k,\theta_2]$
\4          $T[i,i+k,A]:={\bf true}$
\1    for $i\in\{0,1,\ldots,|\alpha|-k\}$
\2      repeat $|G|$ times
\2        \comment for each production $A\to\theta$,
\3        if $T[i,i+k,\theta]$ then $T[i,i+k,A]:={\bf true}$
\0  return $T[0,|\alpha|,S]$
\end{alg}
\smallskip
\caption{
  Pseudocode of the CYK parser variant used in the implementation.
  The parser is for a grammar~$G$ whose starting nonterminal is~$S$.
  We use $T[i,j,\theta]$ to represent $\cyk(\alpha[i:j),\theta)$.
  The pseudocode assumes that the grammar has been preprocessed
    as described in the main text.
  Line~5 is a template: it is repeated for each terminal~$a$;
  lines 10--11 form a template: it is repeated for each binary production;
  line~15 is a template: it is repeated for each unary production.
}\label{fig:cyk-pseudo}
\end{figure}

The pseudocode is given in \autoref{fig:cyk-pseudo}.
On line~1, we check the special case of the empty string.
This is precomputed, perhaps using the memoized version.
In lines 3--5, we handle terminals;
  there is one copy of line~5 for each terminal~$a$.
The loop on line~6 goes over substrings of length~$k$,
  from $k=2$ to $k=|\alpha|$.
Lines 7--11 handle binary productions $A\to\theta_1\theta_2$;
lines 12--15 handle unary productions $A\to\theta$.
There are no other productions, because of the preprocessing we performed.
For unary productions,
  the loop on line~13 is responsible with finding a fixed point.
A fixed point is guaranteed to be reached in $|G|$~iterations,
  because each iteration must change the status of at least one nonterminal,
  but the loop could also stop early if it detects a fixed point.

The variant of CYK we use (\autoref{fig:cyk-pseudo}) runs in
  $\Theta\bigl(|\alpha|^3\cdot|G|+|\alpha|^2\cdot|G|^2\bigr)$ time
  and $\Theta(|\alpha|^2\cdot|G|)$ space.
Moreover,
  the size of the parser's source code is $\Theta\bigl(|G|\bigr)$.
For a plethora of other CYK variants, see \cite{cyk}.

\smallskip

The bounds from above are important for establishing \autoref{thm:fluent}.
First,
  if we want the generated Java classes to be polynomial in the size of the grammar,
  it must be that the size of the parser is polynomial in the size of the grammar.
We showed above that the size of the parser is $O\bigl(|G|\bigr)$.
Second,
  if we want the typechecking of the Java code to be done in polynomial time,
  it must be that the parser works in polynomial time.
We showed above that the parser works in $O\bigl(|\alpha|^3\cdot|G|^2\bigr)$ time.

The reader may ask why not stick to the simpler memoized version,
  instead of going for the traditional dynamic programming version?
The reason is a pragmatic one.
We will implement the parser in a language
  that we will then translate into Turing machines.
To ease the translation to Turing machines,
  we want the intermediate language to be simple.
One way we make it simple is by not having procedure calls.
These are not needed for dynamic programming,
  but they would be needed for memoization.

\medskip

The existing result of \cite{det-cfg-java} handles only context free languages
  that are deterministic.
Intuitively, as explained by \cite{knuth-lr},
  \df{deterministic languages} are those
    that can be parsed in one pass from left to right.
Formally,
  \df{deterministic languages} are those defined by ${\rm LR}(k)$ grammars
  or, equivalently, by ${\rm LR}(1)$ grammars or by deterministic pushdown automata.
\df{Unambiguous grammars} are those for which each string has at most one parse tree.
Unambiguous grammars are a strict superset of ${\rm LR}$ grammars.
The latter could be called prefix-unambiguous:
  not only that strings need to correspond to at most one parse tree,
  but also prefixes must correspond to at most one incomplete parse tree.
A language is said to be \df{inherently ambiguous}
  if there exist grammars that define it but all are ambiguous;
  a standard example \cite{hopull} is
    $\cL_{{\rm ambig}} \;\defeq\; \cL_{()()} \cup \cL_{(())}$
  where
\begin{align*}
  \cL_{()()} \;&\defeq\; \{\,a^mb^mc^nd^n\mid m,n\in\NN\,\} \\
  \cL_{(())} \;&\defeq\; \{\,a^mb^nc^nd^m\mid m,n\in\NN\,\}
\end{align*}
Clearly,
  inherently ambiguous languages are not deterministic and therefore
  cannot be recognized with the approach of \cite{det-cfg-java}.
That is why we shall use $\cL_{{\rm ambig}}$ as an example later on.

\subsection{A Simper Language}\label{sec:simper} 

This section introduces a simple, imperative language,
  which we call \emph{Simper}.
In Simper, we can easily implement CYK parsers.
Furthermore, we can easily compile any Simper program into a Turing machine
  (\autoref{sec:simper->tm}).


\begin{figure*}
\def\l#1{\underline{\sf#1}}
\def\alt{\quad|\quad}
\begin{align*}
  s \;&\to\;
    \ell \l{:}
    \alt \l{goto}\,\ell
    \alt l\l{:=}v
    \alt \l{if}\,c\,\l{\{}s^*\l{\}}\,\bigl(\l{else}\,\l{\{}s^*\l{\}}\bigr)^?
    \alt \l{++}l
    \alt \l{--}l
    \alt \l{halt}
  &&\text{statements (core)}
\\
  s \;&\to\;
    \l{while}\,c\,\l{\{}s^*\l{\}}
    \alt \l{switch}\,v\,\l{\bigl\{}\bigl(v\,\l{\{}s^*\l{\}}\bigr)^*\l{\bigr\}}
  &&\text{statements (sugar)}
\\
  v \;&\to\; l\alt r
  &&\text{values}
\\
  l \;&\to\; {\it id}\bigl(\l{[}v(\l{,}v)^*\l{]}\bigr)^?
  &&\text{left values}
\\
  r \;&\to\;
    {\it nat}
    \alt \l{\text{``}} {\it string} \l{\text{''}}
    \alt \l{array}\l{[}v(\l{,}v)^*\l{]}\l{(}v\l{)}
  &&\text{right values}
\\
  c \;&\to\;
    c\,\l{\tt \&\&}\,c
    \alt c\,\l{\tt ||}\,c
    \alt v\,\l{\tt==}\,v
    \alt v\,\l{\tt!\!=}\,v
  &&\text{conditions}
\end{align*}
\caption{
  Simper syntax.
  The underlined parts appear as given; that is, they are terminals.
  The notation $a\to b\mid c$ is shorthand for $a\to b$ and $a\to c$, as usual.
  The notation $s^*$ is shorthand for `list of zero or more $s$'.
  The notation $s^?$ is shorthand for `zero or one $s$'.
}
\label{fig:simper-syntax}
\end{figure*}

The syntax of Simper is given in \autoref{fig:simper-syntax}.
Minsky machines have counters and zero tests.
The RAM computation model adds an infinite array with store and load statements
  \cite{goldreich}.
Instead,
  Simper adds named arrays of arbitrary size,
  a tiny type system,
  and some syntactic sugar.
These features are added to ease the implementation of parsers,
  such as the one in \autoref{fig:cyk-pseudo}.
Other features, such as procedures and arithmetic operations,
  are not added in order to keep the translation to Turing machines simple.

\begin{figure*}
\begin{align*}
\def\arraystretch{3}
\begin{array}{*4c}
\dfrac%
  {l\;:\;t \qquad v\;:\;t}%
  {l := v \;:\; {\rm unit}}
&
\dfrac%
  {c\;:\;{\rm bool}\qquad b_1\;:\;{\rm unit} \qquad b_2\;:\;{\rm unit}}%
  {{\sf if}\,c\,\{b_1\}\,{\sf else}\,\{b_2\} \;:\; {\rm unit}}
&
\dfrac%
  {l \;:\; {\rm nat}}%
  {\mathop{{+}{+}}l \;:\; {\rm unit}}
&
\dfrac%
  {}%
  {{\sf halt} \;:\; {\rm unit}}
\\
\dfrac%
  {c_1 \;:\; {\rm bool} \qquad c_2 \;:\; {\rm bool}}%
  {c_1\mathop{{\&}{\&}}c_2 \;:\; {\rm bool}}
&
\dfrac%
  {x \;:\; {\rm array}\,n\,t
    \qquad e_1 \;:\; {\rm nat}
    \quad\ldots\quad
    e_n \;:\; {\rm nat}}%
  {x[v_1,\ldots,v_n] \;:\; t}
&
\dfrac%
  {s_1\;:\;{\rm unit} \quad\ldots\quad s_n\;:\;{\rm unit}}%
  {s_1\ldots s_n \;:\; {\rm unit}}
&
\dfrac%
  {}%
  {0 \;:\; {\rm nat}}
\\
\dfrac%
  {v_1 \;:\; t \qquad v_2 \;:\; t}%
  {v_1\mathop{{=}{=}}v_2 \;:\; {\rm bool}}
&
\dfrac%
  {v' \;:\; t
    \qquad v_1 \;:\; {\rm nat}
    \quad\ldots
    \quad v_n \;:\; {\rm nat}}%
  {{\sf array}[v_1,\ldots,v_n](v') \;:\; {\rm array}\,n\,t}
&
\dfrac%
  {}%
  {x \;:\; t}
&
\dfrac%
  {}%
  {\text{``{\tt foo}''} \;:\; {\rm sym}}
\end{array}
\end{align*}
\caption{
  Some representative typing rules of Simper.
}\label{fig:simper-types}
\end{figure*}

The type system of Simper is given in \autoref{fig:simper-types}.
The basic types are
  $\rm nat$ (for nonnegative integers),
  $\rm sym$ (for symbols), and
  $\rm bool$ (for booleans).
The type ${\rm array}\,n\,t$ is the type of
  $n$-dimensional arrays with elements of type~$t$.
We assume a special $1$-dimensional array of symbols called $\it input$,
  whose length is given by a special variable~$n$.
Thus, ${\it input}\;:\;{\rm array}\,1\,{\rm sym}$ and $n\;:\;{\rm nat}$.
Simper programs output only one bit:
  whether they halt (by executing {\sf halt})
  or they get stuck (by reaching the end of the list of statements).

The semantics are straightforward, with a few quirks.
First, since ${\rm nat}$ holds only nonnegative integers,
  $\mathop{{-}{-}}0$ evaluates to~$0$.
Second,
  the notation ${\sf array}[v_1,\ldots,v_n](v')$ stands
  for an array filled with~$v'$ that has $n$~dimensions,
  and the valid indices along the $i$th dimension are $0,\ldots,v_i-1$.

\subsection{Example: Parsers for $\cL_{\rm ambig}$} 

Instead of giving formal semantics for Simper, let us just look at some examples.
The following is a specialized parser for $\cL_{\rm ambig}$:
\begin{Verbatim}
i := 0
a := 0  while i != n && input[i] == "a" { ++i ++a }
b := 0  while i != n && input[i] == "b" { ++i ++b }
c := 0  while i != n && input[i] == "c" { ++i ++c }
d := 0  while i != n && input[i] == "d" { ++i ++d }
if i == n {
  if a == b && c == d { halt }
  if a == d && b == c { halt }
}
\end{Verbatim}
Like in Java, in the expression $c_1\mathop{{\&}{\&}}c_2$,
  we evaluate $c_2$ only if $c_1$~evaluates to false.

Alternatively,
  we can describe $\cL_{\rm ambig}$ using a context free grammar,
  and then we can apply the recipe from \autoref{sec:cyk}.
Let us do this,
  to illustrate what the proof of concept implementation actually generates.

We start with the following grammar:
\begin{align*}
S &\to X
&
X &\to a X d
&
X &\to F
&
Y &\to E G
\\
S &\to Y
&
E &\to a E b
&
F &\to b F c
&
G &\to c G d
\\
&&
E &\to \epsilon
&
F &\to \epsilon
&
G &\to \epsilon
\end{align*}
After limiting the length of the right hand sides to $\le 2$
  and after eliminating $\epsilon$ from the language of each nonterminal,
  we are left with the binary productions
\begin{align*}
Y &\to EG
&
X &\to a X'
&
X' &\to X d
\\
F &\to b F'
&
F' &\to F c
&
E &\to a E'
\\
E' &\to E b
&
G &\to c G'
&
G' &\to G d
\end{align*}
and the unary productions
\begin{align*}
S &\to X
&
X &\to F
&
Y &\to E
&
Y &\to G
\\
S &\to Y
&
E' &\to b
&
F' &\to c
&
G' &\to d
\end{align*}

We can now start implementing the pseudcode from \autoref{fig:cyk-pseudo} in Simper.
Since $\epsilon \in \cL_{\rm ambig}$, we start with a check for this special case:
{\footnotesize
\begin{Verbatim}
  if n == 0 { halt }
\end{Verbatim}
}
Next, we choose some arbitrary injective mapping from symbols to nonnegative integers.
{\tiny
\begin{align*}
\begin{matrix}
  0 & 1 & 2 & 3 & 4 & 5 & 6 & 7 & 8 & 9 & 10 & 11 & 12 & 13 & 14
\\
  S & X & X' & Y & Y' & E & E' & F & F' & G & G' & a & b & c & d
\end{matrix}
\end{align*}}
We can now continue with the initialization of the array~$T$:
{\footnotesize
\begin{Verbatim}
    sn := n   ++sn   T := array[sn,sn,15](0)
    i := 0  si := 1  while i != n {
      switch input[i] {
        "a" { T[i,si,11] := 1 }
        "b" { T[i,si,12] := 1 }
        "c" { T[i,si,13] := 1 }
        "d" { T[i,si,14] := 1 }
      }
      ++i   ++si
    }
\end{Verbatim}
}
\noindent Finally, we have the main loop:
{\footnotesize
\begin{Verbatim}[commandchars=\\\{\},codes={\catcode`$=3}]
    k := 2  while k != sn \{
      i := 0  ik := k  while ik != sn \{
        j := i  ++j  while j != ik \{
          // for $Y \to EG$
          if T[i,j,5] == 1 && T[j,ik,9] == 1 \{
            T[i,ik,3] := 1
          \}
          ... eight other binary productions ...
          ++j
        \}
        ++i   ++ik
      \}
      i := 0  ik := k  while ik != sn \{
        j := 0  while j != 11 \{ // 11 nonterminals
          // for $S \to X$
          if T[i,ik,1] == 1 \{ T[i,ik,0] := 1 \}
          ... seven other unary productions ...
          ++j
        \}
        ++i   ++ik
      \}
    \}
    if T[0,n,0] == 1 \{ halt \}
\end{Verbatim}
}
\noindent
There are three things to notice.
First,
  the lack of arithmetic in Simper forces us to introduce some auxiliary variables,
  such as {\tt ik} which stands for $i+k$.
Overall, however, the inconveniences are minor.
Second,
  these deviations from the pseudocode
  do not affect the size and the running time of the program.
But, third,
  we also notice that the CYK parser is significantly more complicated
  than the specialized parser we started with.

\smallskip

To complete the argument for \autoref{thm:fluent},
  we continue by showing how a Java type checker
  can be used as an interpreter for any Simper program.

\section{Using a Java Compiler as an Interpreter}\label{sec:compiler} 

This section describes a compiler of Simper programs into Java code.
The generated Java code type checks if and only if
  the original Simper program reaches a {\sf halt} statement.
Therefore, a Java type checker can be used as an interpreter.
The compilation is done in two phases:
  a translation from Simper programs into Turing machines
    (\autoref{sec:simper->tm}),
  followed by a translation from Turing machines into Java code
    (\autoref{sec:tm->java}).
The reason for splitting the compilation in two phases
  is that we get to reuse the proof of \autoref{thm:main},
  for the second phase.
We shall use an extended version of Turing machines (\autoref{sec:etm}),
  which slightly complicates the second phase but greatly simplifies the first phase.

The compiler described in this section is part of the proof of \autoref{thm:fluent}.
So, this section has three goals:
(1)~to argue that the compilation is correct,
(2)~to argue that the increase in code size is polynomial only, and
(3)~to argue that programs with a polynomial runtime
  will still have a polynomial runtime after compilation.
The last two points do not guarantee that the compilation is practical,
  but they do suggest that practical compilation can be achieved
  (\autoref{sec:efficiency}).

\subsection{Extended Turing Machines}\label{sec:etm} 

Why do we want to extend Turing machines?
The type ${\rm nat}$ of Simper allows arbitrarily large nonnegative integers.
This means that we do not know in advance
  how many digits are necessary to store the value of a ${\rm nat}$ variable.
In turn, this means that the space allocated for some ${\rm nat}$ variables
  will occasionally need to be extended.
With a normal Turing machine tape,
  one would need to shift the content of roughly half of the tape.
Instead, we will allow extended Turing machines to insert symbols,
  as a primitive operation.
The proof of \autoref{thm:main} can easily be adjusted to account for such insertions.

An \df{extended Turing machine}~$\cE$
  is a tuple $(Q,q_{\sf I},q_{\sf H},\Sigma,\delta)$, where
  $Q$~is a finite set of \df{states},
  $q_{\sf I}$~is the \df{initial state},
  $q_{\sf H}$~is the \df{halt state},
  $\Sigma$~is a finite \df{alphabet},
  and $\delta : Q \times \Sigma_{\bot} \to Q \times \Sigma^* \times \{\dL,\dS,\dR\}$
    is a \df{transition function}.
This definition is the same with that of Turing machines,
  except for the type of the transition function.
A \df{configuration} is a tuple $(q,\alpha,b,\gamma)$
  of the current state~$q$,
  the left part of the tape $\alpha\in\Sigma^*$,
  the current symbol $b\in\Sigma_{\bot}$, and
  the right part of the tape $\gamma\in\Sigma^*$.
The \df{execution steps} of~$\cE$ are the following:
\begin{align*}
  (q,\alpha a,b,\gamma) &\to (q',\alpha,a,\beta\gamma)
    &&\text{for $\delta(q,b)=(q',\beta,\dL)$}
\\
  (q,\alpha,b,\gamma) &\to (q',\alpha,b',\gamma)
    &&\text{for $\delta(q,b)=(q',b',\dS)$}
\\
  (q,\alpha,b,c\gamma) &\to (q',\alpha\beta,c,\gamma)
    &&\text{for $\delta(q,b)=(q',\beta,\dR)$}
\end{align*}
If $\delta(q,b)=(q',\beta,d)$, we require that $d=\dS$ only if $|\beta|=1$.
We also allow for execution steps that go outside the existing tape:
\begin{align*}
  (q,\epsilon,b,\gamma) &\to (q',\epsilon,\bot,\beta\gamma)
    &&\text{for $\delta(q,b)=(q',\beta,\dL)$}
\\
  (q,\alpha,b,\epsilon) &\to (q',\alpha\beta,\bot,\epsilon)
    &&\text{for $\delta(q,b)=(q',\beta,\dS)$}
\end{align*}
As before,
  we require that
    $\delta(q_{\sf H},b)=(q_{\sf H},b,\dS)$ for all $b\in\Sigma$,
    and $\delta(q_{\sf H},\bot)=(q_{\sf H},b,\dS)$ for some $b\in\Sigma$.
A \df{run} on input tape $\alpha_{\sf I}$
  is a sequence of execution steps starting
  from configuration $(q_{\sf I},\epsilon,\bot,\alpha_{\sf I})$.
If $\cE$~reaches $q_{\sf H}$ we say that $\cE$~halts on~$\alpha_{\sf I}$.

\begin{proposition}
We can convert between Turing machines and extended Turing machine
  with only polynomial increases in machine size
  and while preserving polynomial runtimes.
\end{proposition}
\begin{proof}[Proof sketch]
The conversion from Turing machines to extended Turing machines is trivial.
The conversion from extended Turing machines to Turing machines
  can be done by shifting half-tapes whenever the length of the tape increases.
If the original runtime was polynomial,
  then the length of the tape is polynomial
  and so the extra work for shifting is also polynomial.
\end{proof}

\subsection{From Programs to Turing Machines}\label{sec:simper->tm} 

This section describes a compiler of Simper programs into extended Turing machines.
We will see
  (a)~how the Turing tape is organized,
  (b)~what are the main ideas for translating the core Simper statements, and
  (c)~how the compiler itself is organized.

\paragraph{Tape Content.} 

Simper variables hold values of two types, ${\rm nat}$ and ${\rm sym}$.
(No variable has type ${\rm bool}$, which is used for conditions.)
The type ${\rm nat}$~is unbounded: any integer can be incremented.
The type ${\rm sym}$~is bounded:
  we can check if two symbols are equal
  but we cannot combine symbols to produce other symbols.
Values of type ${\rm nat}$ will be represented using bits;
values of type ${\rm sym}$ will be present in the alphabet of the target machine.
So, each constant of type ${\rm sym}$ from the Simper program (such as ``{\tt foo}'')
  will have a letter in the alphabet of the target machine.

In general,
  the alphabet of the target machine is made out of bits, symbols, and markers.
One use of markers is to delimit tape zones that store variable values.
Each variable~$x$ of the Simper program is stored in a tape zone delimited
  by the markers ${\langle_x}$~and~${\rangle_x}$.
The tape starts with zones for the special variables (${\it input}$ and~$n$)
  and continues with zones for the other variables,
  in no particular order.

We also use markers to represent arrays.
Each element of a $1$-dimensional array is delimited
  by markers ${\langle^0}$~and~${\rangle^0}$.
If $a$~is an $n$-dimensional array,
  let us write $a[i]$ for an $(n-1)$-dimensional subarray such that
    $a[i][i_1,\ldots,i_{n-1}] = a[i,i_1,\ldots,i_{n-1}]$
  for all $i_1,\ldots,i_{n-1}$.
Let us also write $\rep{a}$ for the representation of value~$a$.
Then, if $a$~is an $(n+1)$-dimensional array
  whose $k$th coordinate ranges over $0,\ldots,d_k-1$,
  we define
\begin{align}
  \rep{a} &\defeq
    \bigl\langle^n \rep{a[0]} \bigr\rangle^n
    \bigl\langle^n \rep{a[1]} \bigr\rangle^n
    \ldots
    \bigl\langle^n \rep{a[d_0-1]} \bigr\rangle^n
\label{eq:array-rep}
\end{align}
This is a recursive definition.
To specify the base case, we simply regard ${\rm array}\,0\,t$ as isomorphic to~$t$:
\begin{align}
  {\rm array}\,0\,t \;\simeq\; t
  \label{eq:array-dim0}
\end{align}

For example,
  the representation of the $2$-dimensional array
    $\begin{psmallmatrix}0&1\\2&3\end{psmallmatrix}$
  is the following:
\begin{align*}
  \rep{\begin{psmallmatrix}0&1\\2&3\end{psmallmatrix}}
&=
  \langle^1 \rep{(0\;1)} \rangle^1
  \langle^1 \rep{(2\;3)} \rangle^1
  &&\text{by \eqref{eq:array-rep}}
\\&=
  \langle^1
    \langle^0 \rep{(0)} \rangle^0
    \langle^0 \rep{(1)} \rangle^0
  \rangle^1
  \langle^1
    \langle^0 \rep{(2)} \rangle^0
    \langle^0 \rep{(3)} \rangle^0
  \rangle^1
  &&\text{by \eqref{eq:array-rep}}
\\&=
  \langle^1
    \langle^0 \rep{0} \rangle^0
    \langle^0 \rep{1} \rangle^0
  \rangle^1
  \langle^1
    \langle^0 \rep{2} \rangle^0
    \langle^0 \rep{3} \rangle^0
  \rangle^1
  &&\text{by \eqref{eq:array-dim0}}
\\&=
  \langle^1
    \langle^0 0 \rangle^0
    \langle^0 1 \rangle^0
  \rangle^1
  \langle^1
    \langle^0 01 \rangle^0
    \langle^0 11 \rangle^0
  \rangle^1
\end{align*}
On the last line,
  ${\rm nat}$ values are represented in binary,
  with the most significant bit towards the right.
(This bit order, which is used in the implementation, is an arbitrary convention.)


\paragraph{Translation of Statements.} 

Roughly speaking,
  Simper programs can be visualized as flowgraphs
    whose arcs are labeled by statements.
Each vertex of the flowgraph will correspond to a state of the Turing machine.
The translation work consists of transforming each statement arc
  into a set of Turing transitions which have the same effect.

The increment statement $\mathop{{+}{+}}l$ is simulated in two phases.
First, the head of the Turing machine moves to the tape zone designated by~$l$.
Second, 
  a prefix of the form $11\ldots10$ is changed into $00\ldots01$.
If the content of the tape zone designated by~$l$ is all $1$s,
  then they are all changed to $0$s and an extra $1$ is inserted at the end.
We can do such insertions because we use an extended model of Turing machines.
The decrement statement $\mathop{{-}{-}}l$ is simulated analogously.

What is the tape zone designated by~$l$?
If $l$~is a variable~$x$,
  then the tape zone it designates
  is the zone between the markers ${\langle_x}$~and~${\rangle_x}$.
One could go left until $\langle_{\it input}$ is reached,
  and then right until $\langle_x$~is reached.
At this point,
  the head of the Turing machine is at the left extremity of the tape zone we want.
If $l$~is an array element $a[v_0,\ldots,v_{n-1}]$,
  then the tape zone it designates is more difficult to find.
We begin similarly, by finding~$\langle_a$.
Let $i_k$ be the integer to which $v_k$ evaluates.
What we want to do is to
  move right until the $(i_0+1)$th $\langle^{n-1}$ marker is found,
  then move right until $(i_1+1)$th $\langle^{n-2}$ marker is found,
  and so on.
At the end we would have found a marker $\langle^0$.
The tape zone we want is between this marker and its $\rangle^0$~pair.

However, since the values $i_0,\ldots,i_{n-1}$ are unbounded,
  we must keep track of them using the tape.
The easiest way to do so is to introduce auxiliary variables,
  and hence allocate extra tape zones.
We do so in a preprocessing step that works as follows.
Suppose a statement~$s$ contains the subexpression $a[v_0,\ldots,v_{n-1}]$.
We replace the subexpression by $a[x_0,\ldots,x_{n-1}]$,
  where $x_0,\ldots,x_{n-1}$ are fresh variables,
  and we insert before~$s$ assignments $x_k:=v_k$ for each $k$ in $0,\ldots,n-1$.
Because some of the expressions $v_0,\ldots,v_{n-1}$
  might themselves refer to array elements,
  we apply this transformation recursively.
The result is that all array accesses have the form $a[x_0,\ldots,x_{n-1}]$.
Moreover,
  the values stored in $x_0,\ldots,x_{n-1}$ are used for this array access only,
  which means that after the access we can change their values
    without affecting the semantics of the program.
In fact, we will change their value \emph{while} we access the array element.

Finding the tape zone where an array element is stored
  involves several steps of the form
    `move right until the $(i_k+1)$th $\langle^{n-k-1}$ marker'.
We will implement this by executing
  the step `move right until a $\langle^{n-k-1}$ marker' $i_k+1$ times.
How do we repeat $i_k+1$ times?
Recall that, after preprocessing, $i_k$~is the value stored in variable~$x_k$,
  and also recall that we can change the content of variable~$x_k$.
So, we alternate between moving right and decrementing~$x_k$.
This, of course, requires yet another marker
  to remember our position in the array while we decrement~$x_k$.

To simulate assignments $l\mathrel{:=}r$, we proceed as follows.
First, we find the tape zone designated by~$l$,
  we remove its content,
  and we mark its position by~$\downarrow$.
These operations would be cumbersome with traditional Turing machines,
  but are easy with extended Turing machines.
Second, we store the value of~$r$ in the zone marked by~$\downarrow$.
How we do this second step depends on what exactly $r$~is.
If $r$~is a variable or an array access, then it designates a tape zone.
In that case, we identify that tape zone, mark its left side by~$\uparrow$,
  and then copy from~$\uparrow$ to~$\downarrow$ until a $\rangle$ delimiter is found.
If $r$~is a constant of type ${\rm nat}$~or~${\rm sym}$,
  then we simply write that value in the zone marked by~$\downarrow$.
If $r$~is a constant of type ${\rm array}\,n\,t$,
  then we construct an array representation as in~\eqref{eq:array-rep},
  using the trick with auxiliary variables for counting,
    just as we did for locating array elements.
And third, we finish by removing markers, such as $\uparrow$~and~$\downarrow$.

To simulate an ${\sf if}$ statement, we must evaluate conditions.
Let us see how to evaluate an atomic condition $v_1\mathrel{==}v_2$.
The expressions $v_1$~and~$v_2$ can be variables, array accesses, or constants.
Using auxiliary variables, as we did for array accesses,
  we can reduce the problem of evaluating $v_1\mathrel{==}v_2$
  to the problem of evaluating $x_1\mathrel{==}x_2$,
  with $x_1$~and~$x_2$ being variables.
To compare the values of $x_1$~and~$x_2$,
  we mark the left extremities of their tape zones by $\downarrow$~and~$\uparrow$,
  respectively.
Then we repeatedly compare the symbols on the right of the two markers,
  and move the markers to the right, until a $\rangle$~delimiter is reached.
To compare one pair of symbols,
  we must remember one of the symbols while we move from~$\downarrow$ to~$\uparrow$.
Because there is only a finite number of symbols,
  we can encode the required information in the state of the Turing machine.
The other atomic condition, $v_1\mathrel{!\!=}v_2$,
  is simulated analogously.

Composing conditions, labels, and ${\sf goto}$ statements
  do not pose any interesting challenges.

\paragraph{Organization of the Compiler.} 

The compiler is written in OCaml.
Its source code reads like that of an interpreter
  that is written in an imperative style.
For example, the translation of assignments is done by the following function:
{\small\begin{Verbatim}[commandchars=\\\{\}]
\PY{k}{let} \PY{n}{convert\PYZus{}assignment} \PY{n}{left} \PY{n}{right} \PY{o}{:} \PY{n}{transformer} \PY{o}{=}
  \PY{o}{(} \PY{n}{markLvalue} \PY{n}{target\PYZus{}mark} \PY{n}{left}
  \PY{o}{\PYZam{}} \PY{n}{eraseBitsAfter} \PY{n}{target\PYZus{}mark}
  \PY{o}{\PYZam{}} \PY{n}{assign} \PY{n}{right} \PY{o}{)}
\end{Verbatim}
}
\noindent
`Find and mark the tape zone designated by {\tt left};
  then erase its content;
  then fill it with the value of {\tt right}.'
This code closely matches the informal description we saw earlier.
Let us see how this is achieved.

Turing machines are represented essentially as lists of transitions.
{\small\begin{Verbatim}[commandchars=\\\{\}]
\PY{k}{type} \PY{n}{machine} \PY{o}{=}
  \PY{o}{\PYZob{}} \PY{n}{states} \PY{o}{:} \PY{n}{state} \PY{k+kt}{list}
  \PY{o}{;} \PY{n}{alphabet} \PY{o}{:} \PY{n}{letter} \PY{k+kt}{list}
  \PY{o}{;} \PY{n}{transitions} \PY{o}{:} \PY{n}{transition} \PY{k+kt}{list} \PY{o}{\PYZcb{}}
\end{Verbatim}
}\noindent
The compiler makes heavy use of transformers.
{\small\begin{Verbatim}[commandchars=\\\{\}]
\PY{k}{type} \PY{n}{transformer} \PY{o}{=} \PY{o}{(}\PY{n}{state} \PY{o}{*} \PY{n}{machine}\PY{o}{)} \PY{o}{\PYZhy{}\PYZgt{}} \PY{o}{(}\PY{n}{state} \PY{o}{*} \PY{n}{machine}\PY{o}{)}
\end{Verbatim}
}\noindent
Intuitively, if $t$ is a transformer,
  then calling $t(q,\cT)$
    will add some transitions emanating from state~$q$
    and eventually converging in a fresh state~$q'$.
The call $t(q,\cT)$ then returns a pair $(q',\cT')$,
  where $\cT'$~is a machine obtained by adding the extra transitions to~$\cT$.
Suppose we have a transformer $t_1$ that simulates the Simper statement~$s_1$,
  and a transformer $t_1$ that simulates the Simper statement~$s_2$.
If we want to simulate executing $s_1$ and then $s_2$,
  then we need to apply $t_1$ and then $t_2$ to the Turing machine being built.
To compose transformers, we use the following straightforward combinators:
{\small\begin{Verbatim}[commandchars=\\\{\}]
\PY{k}{let} \PY{o}{(} \PY{o}{\PYZam{}} \PY{o}{)} \PY{n}{f} \PY{n}{g} \PY{n}{x} \PY{o}{=} \PY{n}{g} \PY{o}{(}\PY{n}{f} \PY{n}{x}\PY{o}{)}
\PY{k}{let} \PY{n}{seqs} \PY{o}{=} \PY{n+nn}{List}\PY{p}{.}\PY{n}{fold\PYZus{}left} \PY{o}{(} \PY{o}{\PYZam{}} \PY{o}{)} \PY{o}{(}\PY{k}{fun} \PY{n}{x}\PY{o}{\PYZhy{}\PYZgt{}}\PY{n}{x}\PY{o}{)}
\end{Verbatim}
}\noindent
We can now see that the definition of {\tt convert\_assignment}
  simply composes three transformers.

To implement ${\sf if}$ statements,
  we also make use of three new types of transformers:
  {\tt branch}, {\tt join}, and {\tt transformer2}.
{\small\begin{Verbatim}[commandchars=\\\{\}]
\PY{k}{type} \PY{n}{state2} \PY{o}{=} \PY{o}{\PYZob{}} \PY{n}{yes\PYZus{}branch} \PY{o}{:} \PY{n}{state}\PY{o}{;} \PY{n}{no\PYZus{}branch} \PY{o}{:} \PY{n}{state} \PY{o}{\PYZcb{}}
\PY{k}{type} \PY{n}{branch} \PY{o}{=} \PY{n}{state} \PY{o}{*} \PY{n}{machine} \PY{o}{\PYZhy{}\PYZgt{}} \PY{n}{state2} \PY{o}{*} \PY{n}{machine}
\PY{k}{type} \PY{n}{join} \PY{o}{=} \PY{n}{state2} \PY{o}{*} \PY{n}{machine} \PY{o}{\PYZhy{}\PYZgt{}} \PY{n}{state} \PY{o}{*} \PY{n}{machine}
\PY{k}{type} \PY{n}{transformer2} \PY{o}{=} \PY{n}{state2} \PY{o}{*} \PY{n}{machine} \PY{o}{\PYZhy{}\PYZgt{}} \PY{n}{state2} \PY{o}{*} \PY{n}{machine}
\end{Verbatim}
}\noindent
Such transformers are useful whenever we need to take decisions.
For example, when we copy a value from one variable to another,
  we need to stop when we see a $\rangle$~delimiter.
For taking such decisions, we use branching and joining transformers.
For a concrete and simple example of how such transformers are used,
  let us look at how the ${\sf if}$ statement is translated.
{\small\begin{Verbatim}[commandchars=\\\{\}]
\PY{k}{let} \PY{n}{convert\PYZus{}if} \PY{n}{condition} \PY{n}{ys} \PY{n}{ns} \PY{o}{:} \PY{n}{transformer} \PY{o}{=}
  \PY{o}{(} \PY{n}{branch\PYZus{}on} \PY{n}{condition}
    \PY{o}{\PYZam{}} \PY{n}{yes\PYZus{}branch} \PY{o}{(}\PY{n}{convert\PYZus{}body} \PY{n}{ys}\PY{o}{)}
    \PY{o}{\PYZam{}} \PY{n}{no\PYZus{}branch} \PY{o}{(}\PY{n}{convert\PYZus{}body} \PY{n}{ns}\PY{o}{)}
  \PY{o}{\PYZam{}} \PY{n}{join} \PY{o}{)}
\end{Verbatim}
}\noindent
The types involved are the following:
{\small\begin{Verbatim}[commandchars=\\\{\}]
\PY{k}{val} \PY{n}{branch\PYZus{}on} \PY{o}{:} \PY{n}{condition} \PY{o}{\PYZhy{}\PYZgt{}} \PY{n}{branch}
\PY{k}{val} \PY{n}{yes\PYZus{}branch} \PY{o}{:} \PY{n}{transformer} \PY{o}{\PYZhy{}\PYZgt{}} \PY{n}{transformer2}
\PY{k}{val} \PY{n}{no\PYZus{}branch} \PY{o}{:} \PY{n}{transformer} \PY{o}{\PYZhy{}\PYZgt{}} \PY{n}{transformer2}
\PY{k}{val} \PY{n}{join} \PY{o}{:} \PY{n}{join}
\end{Verbatim}
}\noindent
The combinator {\tt yes\_branch} is implemented as follows:
{\small\begin{Verbatim}[commandchars=\\\{\}]
\PY{k}{let} \PY{n}{yes\PYZus{}branch} \PY{n}{t} \PY{o}{=} \PY{k}{fun} \PY{o}{(}\PY{o}{\PYZob{}} \PY{n}{yes\PYZus{}branch}\PY{o}{;} \PY{n}{no\PYZus{}branch} \PY{o}{\PYZcb{}}\PY{o}{,} \PY{n}{tm}\PY{o}{)} \PY{o}{\PYZhy{}\PYZgt{}}
  \PY{k}{let} \PY{n}{yes\PYZus{}branch}\PY{o}{,} \PY{n}{tm} \PY{o}{=} \PY{n}{t} \PY{o}{(}\PY{n}{yes\PYZus{}branch}\PY{o}{,} \PY{n}{tm}\PY{o}{)} \PY{k}{in}
  \PY{o}{(}\PY{o}{\PYZob{}} \PY{n}{yes\PYZus{}branch}\PY{o}{;} \PY{n}{no\PYZus{}branch} \PY{o}{\PYZcb{}}\PY{o}{,} \PY{n}{tm}\PY{o}{)}
\end{Verbatim}
}\noindent
The combinator {\tt no\_branch} is implemented analogously.
The (join) transformer {\tt join} adds a fresh state and no-op transitions to it
  from both of the given states.
The (branch) transformer builder {\tt branch\_on} is more involved,
  since it depends on {\tt condition}.
The function {\tt convert\_body} is defined as follows:
{\small\begin{Verbatim}[commandchars=\\\{\}]
\PY{k}{let} \PY{k}{rec} \PY{n}{convert\PYZus{}body} \PY{o}{(}\PY{n}{xs} \PY{o}{:} \PY{n}{statement} \PY{k+kt}{list}\PY{o}{)} \PY{o}{:} \PY{n}{transformer} \PY{o}{=}
  \PY{n}{seqs} \PY{o}{(}\PY{n+nn}{List}\PY{p}{.}\PY{n}{map} \PY{n}{convert\PYZus{}statement} \PY{n}{xs}\PY{o}{)}
\PY{o+ow}{and} \PY{n}{convert\PYZus{}statement} \PY{o}{(}\PY{n}{x} \PY{o}{:} \PY{n}{statement}\PY{o}{)} \PY{o}{:} \PY{n}{transformer} \PY{o}{=}
  \PY{c}{(*}\PY{c}{ ... uses convert\PYZus{}body ... }\PY{c}{*)}
\end{Verbatim}
}\noindent
where {\tt seqs}, defined earlier, takes a list of transformers
  and chains them sequentially.

Transformers and combinators for transformers --
  this is the mortar used to build the compiler.
The bricks are transformers such as
  {\tt goLeft}, {\tt goToMark}, {\tt insertOnRight}, {\tt write},
  and many others.


\subsection{From Turing Machines to Java}\label{sec:tm->java} 

The next step is to compile Turing machines into Java code.
This is achieved by a tiny compiler,
  which implements the proof of \autoref{thm:main},
  with small adaptations to allow for our extensions (\autoref{sec:etm}).

The difference between normal and extended Turing machines
  can be seen in the transition function type:
\begin{align*}
  &\text{normal type}
&
  Q\times\Sigma_\bot &\to Q\times\Sigma\times\{-1,0,+1\}
\\
  &\text{extended type}
&
  Q\times\Sigma_\bot &\to Q\times\Sigma^{\color{darkred}*}\times\{-1,0,+1\}
\end{align*}
Consider an extended Turing machine such that $|\beta|=1$
  whenever $\delta(q,b)=(q',\beta,d)$.
It is easy to see that such an extended machine trivially corresponds to a normal one.
Accordingly,
  we translate such an extended machine
  in exactly the same way as we would translate the normal one.
We only need to do something different for those transitions with $|\beta|\ne1$.
These differences are illustrated in \autoref{fig:reduction-extended-delta}.
Briefly, each $\beta=b_0\ldots b_{n-1}$ in the transition function
  results in a string $L_{b_0}N\ldots L_{b_{n-1}}N$ of classes
  in the corresponding inheritance rule.
The special case $|\beta|=1$ is identical
  to the translation for non-extended transition functions.

\begin{figure*} 
\def\CR#1{{\color{darkred}#1}}
\def\CG#1{{\color{darkgreen}#1}}
\def\CB#1{{\color{darkblue}#1}}
\begin{align*}
\CR{Q_s^{\sf L}}x &\inh
  \CR{L_b N}
  \CG{Q_{s'}^{\sf wL}}
  \CB{M^{\sf L} N}
  \CG{L_{b_0} N L_{b_1} N \ldots L_{b_{n-1}} N} x
  &&\text{for $\delta(\CR{q_s},\CR{b})
    =(\CG{q_{s'}},\CG{b_0b_1\ldots b_{n-1}},\CB{\dL})$}
\\
\CR{Q_s^{\sf L}} x &\inh
  \CR{L_b N}
  \CG{Q_{s'}^{\sf wL}}
  \CB{M^{\sf R} N}
  \CG{L_{b'} N}
  x
  &&\text{for $\delta(\CR{q_s},\CR{b})=(\CG{q_{s'}},\CG{b'},\CB{\dS})$}
\\
\CR{Q_s^{\sf L}}x &\inh
  \CR{L_b N}
  \CG{Q_{s'}^{\sf wL}}
  \CG{L_{b_0} N L_{b_1} N \ldots L_{b_{n-1}} N}
  \CB{M^{\sf R} N}
  x
  &&\text{for $\delta(\CR{q_s},\CR{b})
    =(\CG{q_{s'}},\CG{b_0b_1\ldots b_{n-1}},\CB{\dR})$}
\\[0.5ex]
\CR{Q_s^{\sf L}}x &\inh
  \CR{L_\# N}
  \CG{Q_{s'}^{\sf wL}}
  \CG{L_\# N}
  \CB{M^{\sf L} N}
  \CG{L_{b_0} N L_{b_1} N \ldots L_{b_{n-1}} N}
  x
  &&\text{for $\delta(\CR{q_s},\CR{\bot})
    =(\CG{q_{s'}},\CG{b_0b_1\ldots b_{n-1}},\CB{\dL})$}
\\
\CR{Q_s^{\sf L}} x &\inh
  \CR{L_\# N}
  \CG{Q_{s'}^{\sf wL}}
  \CG{L_\# N}
  \CB{M^{\sf R} N}
  \CG{L_{b'} N}
  x
  &&\text{for $\delta(\CR{q_s},\CR{\bot})=(\CG{q_{s'}},\CG{b'},\CB{\dS})$}
\\
\CR{Q_s^{\sf L}}x &\inh
  \CR{L_\# N}
  \CG{Q_{s'}^{\sf wL}}
  \CG{L_\# N}
  \CG{L_{b_0} N L_{b_1} N \ldots L_{b_{n-1}} N}
  \CB{M^{\sf R} N}
  x
  &&\text{for $\delta(\CR{q_s},\CR{\bot})
    =(\CG{q_{s'}},\CG{b_0b_1\ldots b_{n-1}},\CB{\dR})$}
\end{align*}
\caption{
  Class table used to simulate an \emph{extended} Turing machine
    $\cE=(Q,q_{\sf I},q_{\sf H},\Sigma,\delta)$.
  There are six more rules obtained from the ones above
    by swapping ${\sf L}\leftrightarrow{\sf R}$.
  The inheritance rules not involving the transition function~$\delta$
    are the same as in \autoref{fig:reduction}.
}
\label{fig:reduction-extended-delta}
\end{figure*}

\smallskip

This concludes our translation from Simper programs to Java code.
The Simper program halts if and only if the Java code type checks,
  on all inputs.
For Simper programs, the input is provided in the special array ${\it input}$.
For Java code, the input is provided by a builder (\autoref{sec:builders}).

\subsection{Discussion on Efficiency}\label{sec:efficiency} 

The previous sections (\autoref{sec:simper->tm}~and~\autoref{sec:tm->java})
  describe a translation from Simper programs to Java code.
Is it efficient?
In theory, yes; in practice, no.
We consider two aspects of the generated Java code:
(1)~its size, and
(2)~how fast it can be type checked.

\paragraph{In theory.} 

The Java code has size quadratic in the size of the source Simper program.
More precisely,
  the translation from Simper programs to Turing machines is linear,
  while the translation from Turing machines to Java code is quadratic.
Let us see why.

To simulate increments, decrements, copying and comparisons of values,
  we put, in the Turing machine,
  loops that go over the bits or symbols of a variable or a couple of variables.
Such a loop requires only a constant number of states and transitions.
The most involved construction we had
  was for identifying the tape zone of an array element.
That construction was essentially a sequence of $n$~loops,
  where $n$~is the dimension of the array.
In all cases,
  the part of a Turing machine that simulates a statement~$s$
  has a size proportional to the size of~$s$.
In particular,
  referring to an element of an $n$-dimensional array
    requires $\Theta(n)$ characters in the Simper program,
  and $\Theta(n)$ states and transitions in the Turing machine.
In addition,
  our construction can be done with a transition function such that $|\beta|\le2$
  whenever $\delta(q,b)=(q',\beta,d)$.

From the Turing machine, we generate two pieces of Java code:
  (i)~interface declarations corresponding to \autoref{fig:reduction-extended-delta},
  and (ii)~a builder class as in \autoref{fig:builder}.
By examining these figures,
  we see that the interface declarations will have
    size $\Theta\bigl(|Q|\cdot|\Sigma|\bigr)$,
  and that the builder class will have
    size $\Theta\bigl(|\Sigma|\bigr)$.
Speaking roughly, we can describe this as being quadratic
  in the size of the Turing machine.

\smallskip

Let us now analyze the speed of the simulation.
We start with a Simper program whose runtime is~$t$,
  possibly depending on the size of the input.
We end up with Java code that specifies a subtyping machine,
  whose runtime is~$t'$.
The subtyping machine simulates the Simper program.
The goal in what follows is to put on upper bound on~$t'$, as a function of~$t$.

To find a bound,
  we need to be more precise about how to measure resource usage
  of Simper programs, Turing machines, and subtyping machines.
Let us start with Simper programs.
We consider that each value of type ${\rm nat}$ or ${\rm sym}$ takes $1$~memory cell.
The space taken by an array is obtained by adding the space taken by its elements.
For time, we count each core statement as taking $1$~step,
  with one exception.
The exception is the creation of arrays as a result of using an array literal
  ${\sf array}[v_0,\ldots,v_{n-1}](v')$:
  the time for creating such an array is the space occupied by the array.
With these rules, we see that,
  if the runtime of a Simper program is~$t$, then the space it uses is~$\le t$.
Also, since ${\rm nat}$ values can only be incremented,
  all such values are at most~$O(t)$.

Now let us move on to extended Turing machines.
We define the size of a configuration $(q,\alpha,b,\gamma)$ as
  $\lg|Q| + |\alpha b \gamma|$.
The time taken to execute a Turing machine is the distance in the configuration graph
  from the initial configuration to a final configuration.
We organized the tape into zones, one for each variable in the Simper program,
  plus a few zones for auxiliary variables.
Each type zone for a value of type ${\rm sym}$ has size $O(1)$;
each type zone for a value of type ${\rm nat}$ has size $O(\log t)$;
each type zone for a value of type ${\rm array}$ has size $O(t\log t)$.
All auxiliary variables have type ${\rm nat}$ or ${\rm sym}$,
  and the number of auxiliary variables is at most~$O(p)$,
  where $p$~is the size of the text of the Simper program.
In all, the tape of the Turing machine has size at most $O(pt\log t)$.
What about time?
Each step of a Simper program is simulated by a constant number of Turing transitions.
These transitions may loop over the tape symbols used to represent
  one Simper value.
In the case of ${\rm nat}$ values, this gives an overhead of $O(\log t)$.
For operations such as comparison,
  the head of the Turing machine need to move from one tape zone to another,
  for another overhead of $O(pt\log t)$.
In all, a Simper program that runs in $t$~time steps
  is simulated by an extended Turing machine in $O(pt^2\log^2t)$ time steps.

Now let us move on to subtyping machines.
For these we measure space and time as we do for extended Turing machine:
  space is the size of the tape;
  time is the number of transitions in the configuration graph.
As explained in \autoref{sec:result} (\autoref{fig:configs}),
  there is a direct correspondence between
    configurations of Turing machines and those of subtyping machines.
Thus, the space used by subtyping machines is also $O(pt\log t)$.
For time,
  we note that a subtyping machine head may move back-and-forth once
  in order to simulate one step of the Turing machine.
Thus, the time spent in the subtyping machine is $O(p^2t^3\log^3t)$.

\smallskip

Let us see how the slowdown affects one particular type of Simper programs:
  the CYK parsers we used earlier (\autoref{sec:cyk}).
In that case,
  $t\in O\bigl(|\alpha|^3\cdot |G|^2\bigr)$
  and
  $p\in O\bigl(|G|\bigr)$.
Therefore, the subtyping machine runs in time
  $\tilde O\bigl(|\alpha|^9\cdot |G|^8\bigr)$.
This establishes the polynomial time promised in \autoref{thm:fluent}.
However, the degree of the polynomial is not exactly encouraging.
There is room for improvement,
  and work to be done to achieve a practical parser generator for fluent interfaces.

\paragraph{In practice.} 

Let us consider the example of recognizing the language $\cL_{\rm ambig}$.
With the current implementation,
  the generated Java code is over $1$~GiB,
  and is too big for {\sf javac} to handle in reasonable time.
However,
  if we hand-craft a Turing machine that recognizes $\cL_{\rm ambig}$,
  then the generated Java code is only a few MiB in size.
It is then possible to use {\sf javac} to recognize $\cL_{\rm ambig}$,
  with strings of up to ten letters being handled in minutes.

The main bottleneck now is the compiler from Simper to Turing machines.
To improve its performance, one may extend Turing machines further;
  for example, with a `go to marker' primitive operation.
Or, perhaps one might find a reduction that does not use Turing machines
  as an intermediate representation.


\section{Discussion}\label{sec:discussion} 

Does the reduction from \autoref{sec:result} apply to other languages
  like C${}^\sharp$ and Scala?
No.
Both of them adopted the recursive--expansive restriction of
  \cite{expansive-recursive},
  as recommended by \cite{kennedy2007}.
Roughly,
  this restriction is a syntactic check that succeeds
  if and only if
  our Turing tapes are bounded.
For bounded tapes, the halting problem is decidable.

\medskip

What is the practical relevance of \autoref{thm:main}?
For C, there exists a formally verified type checker in CompCert
  \cite[version~2.5]{compcert}.
For Java,
  \autoref{thm:main} implies that a formally verified type checker
  that guarantees partial correctness cannot also guarantee termination.
In most applications, partial correctness suffices, but not so in
  security critical applications (where users are malicious),
  nor in mission critical applications (where nontermination is costly).
Since there cannot be a totally correct Java type checker,
  \autoref{thm:main}
  strengthens the motivation behind research into alternative type systems,
  such as \cite{f-bounded} and \cite{lightweight-generics}.

It is perhaps difficult to imagine how could the Java type checker
  be the target of a security attack.
However, consider a scenario in which generics are reified,
  as has been discussed for a future version of Java.
In that case,
  one needs to perform subtype checks at runtime to implement {\tt instanceof}.
In other words,
  {\tt instanceof} becomes a potential security vulnerability
  for any Java program that uses it.
A simple fix might be to change the specification of {\tt instanceof}
  to allow it to throw a timeout exception.
Let us hope that a better solution exists.
Similar problems occur if one tries to add gradual typing to Java.

\smallskip
What is the practical relevance of \autoref{thm:fluent}?
On the one hand,
  it hints that it may be possible to implement parser generators
    for fluent interfaces,
  which would make it easier to embed domain specific languages in Java.
On another hand,
  the techniques used in its proof may be reusable
  for encoding other computations in Java's type system.

\medskip

To summarize,
  Theorems \ref{thm:main}~and~\ref{thm:fluent}
  have practical implications:
\begin{enumerate}
\item
A formally verified type checker for Java
  can guarantee partial correctness
  but cannot guarantee termination.
\item
If reified generics are added to Java,
  then one needs to find some solution for not turning {\tt instanceof}
  into a security problem.
Similarly, any technique (such as gradual typing)
  that involves subtype checking at runtime is a potential security problem.
\item
It may be possible to develop parser generators
  that help with embedding domain specific languages in Java.
\end{enumerate}

\smallskip
Still, \autoref{thm:main} is primarily of theoretical interest:
  It strengthens the best lower bound on Java type checking \cite{det-cfg-java},
  and finally answers a question posed almost a decade ago by \cite{kennedy2007}.

\section{Related Work} 

Java's type system is not only undecidable, but also unsound \cite{javats-unsound}.
Many other languages have undecidable type systems:
  Haskell with extensions \cite{haskell-ts-universal,haskell-funcdeps},
  OCaml \cite{lillibridge-phd,ocaml-undecidable},
  C++ \cite{c++-ts-universal},
  and Scala \cite{scala-ts-universal,scala-unsound-fix},
  to name a few.
It is often the case that undecidable type systems
  are a fun playground for metaprogramming.
For example,
  in Haskell, undecidable instances are useful for deriving type classes
    \cite{replib};
  and, in C++, libraries like Boost.Hana \cite{boost-hana}
    allow one to easily write (mostly) functional programs that run at compile time.
Will Java's type system become a playground for metaprogramming?
\cite{det-cfg-java} and \autoref{thm:fluent} show that it is possible in principle.
It remains to be seen whether it is possible in practice.

On the theory side,
  perhaps the closest result
  is that subtype checking is undecidable in~$F_{\subtype}$
  \cite{f-sub-undecidable}.
That proof has some high level similarities with the one presented here:
  it considers a deterministic type system,
  and shows a reduction from (a kind of) Turing machines,
  by using an intermediate formalism.
We also know that System~$F$ is undecidable \cite{system-f-undecidable}.
But, its subset known as Hindley--Milner is decidable,
  in fact type inference is {\sc DExpTime}-complete
  \cite{hm-dexptime-complete}.

Why should we care about toy type systems like $F_{\subtype}$ and that of System~$F$?
Because real type systems like Java's are too difficult for humans to reason about.
There are three ways in which one can settle the decidability of type checking
  for a real world language~X:
(i)~find a simple type system that is less expressive than that of~X,
  and show it is undecidable;
(ii)~find a simple type system that is more expressive than that of~X,
  and show that it is decidable; or
(iii)~tackle the original type system but do a mechanically verified proof.
For~(iii), a huge effort is required.
This paper falls in category~(i).
The simple type system is neither $F_{\subtype}$ nor System~$F$.
Instead, it is taken from \cite{kennedy2007}.
Kennedy and Pierce study several variations of the type system they propose.
Alas, those variations they prove decidable are less expressive than Java,
  and those variations they prove undecidable are not less expressive than Java.
The variant they prove undecidable allows classes of arity~$2$
  and multiple instantiation inheritance,
  which in our setting translates to having nondeterministic subtyping machines.
Their reduction is from PCP\null.
They conjecture that multiple instantiation inheritance
  is essential for undecidability.
That turns out not to be the case.

\section{Conclusion} 

It is possible to coerce Java's type checker into performing \emph{any} computation.
This opens possibilities for use and abuse.

\acks 
The problem was brought to my attention by \cite{cstheory-surprising},
  and it started to look very interesting after I read the result of
    \cite{det-cfg-java}.
Stefan Kiefer, Rasmus Lerchedahl Petersen, Joshua Moerman, Ross Tate,
  and Brian Goetz
  provided feedback on earlier drafts.
Nada Amin answered my questions about Scala.
Hongseok Yang encouraged me to write a proper paper.

\bibliographystyle{abbrvnat}

\raggedright
\bibliography{javats}

\begin{thebibliography}{31}
\providecommand{\natexlab}[1]{#1}
\providecommand{\url}[1]{\texttt{#1}}
\expandafter\ifx\csname urlstyle\endcsname\relax
  \providecommand{\doi}[1]{doi: #1}\else
  \providecommand{\doi}{doi: \begingroup \urlstyle{rm}\Url}\fi

\bibitem[Amin and Tate(2016)]{javats-unsound}
N.~Amin and R.~Tate.
\newblock {Java} and {Scala}'s type systems are unsound: The existential crisis
  of null pointers.
\newblock In \emph{OOPSLA}, 2016.

\bibitem[Bjarnason(2009)]{scala-ts-universal}
R.~O. Bjarnason.
\newblock More {Scala} typehackery.
\newblock \url{https://apocalisp.wordpress.com/2009/09/02/} (see also
  \url{https://michid.wordpress.com/2010/01/29}), 2009.

\bibitem[Breslav(2013)]{kotlin-call}
A.~Breslav.
\newblock Mixed-site variance in {Kotlin}.
\newblock \url{http://blog.jetbrains.com/kotlin/2013/06/mixed}, 2013.

\bibitem[Courcelle(1977)]{courcelle}
B.~Courcelle.
\newblock On jump-deterministic pushdown automata.
\newblock \emph{Mathematical Systems Theory}, 1977.

\bibitem[Dionne(2016)]{boost-hana}
L.~Dionne.
\newblock {Boost.Hana}.
\newblock \url{http://boostorg.github.io/hana}, 2016.

\bibitem[Gil and Levy(2016)]{det-cfg-java}
Y.~Gil and T.~Levy.
\newblock Formal language recognition with the {Java} type checker.
\newblock In \emph{ECOOP}, 2016.

\bibitem[Goldreich(2008)]{goldreich}
O.~Goldreich.
\newblock \emph{Computational complexity: a conceptual perspective}.
\newblock Cambridge University Press, 2008.

\bibitem[Gosling et~al.(2015)Gosling, Joy, Steele, Bracha, and
  Buckley]{java-ls}
J.~Gosling, B.~Joy, G.~Steele, G.~Bracha, and A.~Buckley.
\newblock The {Java} language specification, 2015.
\newblock Java SE 8 Edition.

\bibitem[Greenman et~al.(2014)Greenman, Muehlboeck, and Tate]{f-bounded}
B.~Greenman, F.~Muehlboeck, and R.~Tate.
\newblock Getting {F}-bounded polymorphism into shape.
\newblock In \emph{PLDI}, 2014.

\bibitem[Grigore(2016)]{website}
R.~Grigore.
\newblock Parser generator for fluent interfaces.
\newblock \url{http://rgrig.appspot.com/javats}, 2016.

\bibitem[Hopcroft and Ullman(1979)]{hopull}
J.~E. Hopcroft and J.~D. Ullman.
\newblock \emph{Introduction to Automata Theory, Languages, and Computation}.
\newblock Addison--Wesley, 1979.

\bibitem[Kennedy and Pierce(2007)]{kennedy2007}
A.~J. Kennedy and B.~Pierce.
\newblock On decidability of nominal subtyping with variance.
\newblock In \emph{FOOL}, 2007.

\bibitem[Knuth(1965)]{knuth-lr}
D.~E. Knuth.
\newblock On the translation of languages from left to right.
\newblock \emph{Information and Control}, 1965.

\bibitem[Lange and Lei{\ss}(2009)]{cyk}
M.~Lange and H.~Lei{\ss}.
\newblock To {CNF} or not to {CNF}? {An} efficient yet presentable version of
  the {CYK} algorithm.
\newblock \emph{Informatica Didactica}, 2009.

\bibitem[Leroy(2009)]{compcert}
X.~Leroy.
\newblock Formal verification of a realistic compiler.
\newblock \emph{CACM}, 2009.

\bibitem[Lillibridge(1997)]{lillibridge-phd}
M.~Lillibridge.
\newblock \emph{Translucent Sums: A Foundation for Higher-Order Modules}.
\newblock PhD thesis, CMU, 1997.

\bibitem[Lippert(2013)]{cstheory-surprising}
E.~Lippert.
\newblock A simple problem whose decidability is not known.
\newblock Theoretical Computer Science Stack Exchange, 2013.
\newblock \url{http://cstheory.stackexchange.com/q/18866}.

\bibitem[Mairson(1990)]{hm-dexptime-complete}
H.~G. Mairson.
\newblock Deciding {ML} typability is complete for deterministic exponential
  time.
\newblock In \emph{POPL}, 1990.

\bibitem[Odersky(2016)]{scala-unsound-fix}
M.~Odersky.
\newblock Scaling {DOT} to {Scala} -- soundness.
\newblock
  \url{http://scala-lang.org/blog/2016/02/17/scaling-dot-soundness.html}, 2016.

\bibitem[Pierce(1994)]{f-sub-undecidable}
B.~C. Pierce.
\newblock Bounded quantification is undecidable.
\newblock \emph{Information and Computation}, 1994.

\bibitem[Rossberg(1999)]{ocaml-undecidable}
A.~Rossberg.
\newblock Undecidability of {OCaml} type checking.
\newblock Caml mailing list, 1999.
\newblock
  \hbox{\url{http://caml.inria.fr/pub/old_caml_site/caml-list/1507.html}}.

\bibitem[Sulzmann et~al.(2007)Sulzmann, Duck, Jones, and
  Stuckey]{haskell-funcdeps}
M.~Sulzmann, G.~J. Duck, S.~L.~P. Jones, and P.~J. Stuckey.
\newblock Understanding functional dependencies via constraint handling rules.
\newblock \emph{J. Funct. Program.}, 2007.

\bibitem[Tate et~al.(2011)Tate, Leung, and Lerner]{taming-wildcards}
R.~Tate, A.~Leung, and S.~Lerner.
\newblock Taming wildcards in {Java}'s type system.
\newblock In \emph{PLDI}, 2011.

\bibitem[Turing(1936)]{turing}
A.~M. Turing.
\newblock On computable numbers, with an application to the
  {Entscheidungsproblem}.
\newblock \emph{J. of Math}, 1936.

\bibitem[Veldhuizen(2003)]{c++-ts-universal}
T.~L. Veldhuizen.
\newblock C++ templates are {Turing} complete.
\newblock Technical report, Indiana University, 2003.

\bibitem[Viroli(2000)]{expansive-recursive}
M.~Viroli.
\newblock On the recursive generation of parametric types.
\newblock Technical report, University of Bologna, 2000.

\bibitem[Wansbrough(1998)]{haskell-ts-universal}
K.~Wansbrough.
\newblock Instance declarations are universal.
\newblock \url{http://www.lochan.org/keith/publications/undec.html} (see also
  \url{https://wiki.haskell.org/Type_SK}), 1998.

\bibitem[Wehr and Thiemann(2009)]{wehr2009}
S.~Wehr and P.~Thiemann.
\newblock On the decidability of subtyping with bounded existential types.
\newblock In \emph{APLAS}, 2009.

\bibitem[Weirich(2006)]{replib}
S.~Weirich.
\newblock {RepLib}: a library for derivable type classes.
\newblock In \emph{Haskell}, 2006.

\bibitem[Wells(1999)]{system-f-undecidable}
J.~B. Wells.
\newblock Typability and type checking in {System~$F$} are equivalent and
  undecidable.
\newblock \emph{Annals of Pure and Applied Logic}, 1999.

\bibitem[Zhang et~al.(2015)Zhang, Loring, Salvaneschi, Liskov, and
  Myers]{lightweight-generics}
Y.~Zhang, M.~C. Loring, G.~Salvaneschi, B.~Liskov, and A.~C. Myers.
\newblock Lightweight, flexible object-oriented generics.
\newblock In \emph{PLDI}, 2015.

\end{thebibliography}

\end{document}